\documentclass[a4paper,fleqn]{cas-sc}

\usepackage{amsmath,amsfonts,amssymb}
\usepackage{algorithmic}
\usepackage{algorithm}
\usepackage{array}
\newcolumntype{P}[1]{>{\centering\arraybackslash}p{#1}}
\usepackage[caption=false,font=normalsize,labelfont=sf,textfont=sf]{subfig}
\usepackage{textcomp}
\usepackage{stfloats}
\usepackage{url}
\usepackage{verbatim}
\usepackage{graphicx}
\usepackage{booktabs}
\usepackage{multirow}
\usepackage{siunitx}
\usepackage[table]{xcolor}
\usepackage{color,soul}
\sethlcolor{yellow}
\usepackage{mathrsfs}
\usepackage{tikz}
\usepackage{pgfplots}
\pgfplotsset{compat=1.18}
\usetikzlibrary{arrows.meta,positioning,calc}
\usepackage{threeparttable}
\usepackage{makecell}
\usepackage{adjustbox}
\usepackage{float}

\usepackage[numbers]{natbib}

\def\BibTeX{{\rm B\kern-.05em{\sc i\kern-.025em b}\kern-.08em
T\kern-.1667em\lower.7ex\hbox{E}\kern-.125emX}}

\newtheorem{defka}{Definition}
\newtheorem{thm}{Theorem}
\newtheorem{assumption}{Assumption}

\newtheorem{lem}{Lemma}
\newtheorem{rmk}{Remark}

\newtheorem{pf}{Proof}

\begin{document}

\let\WriteBookmarks\relax
\def\floatpagefraction{1}
\def\textpagefraction{.001}


\title [mode = title] {Adaptive Modular Geometric Control of Robotic Manipulators}

\author[1]{Mahdi Hejrati} \cormark[1] 
\ead{mahdi.hejrati@tuni.fi} 

\author[1]{Amir Hossein Barjini}
\ead{amirhossein.barjini@tuni.fi}

\author[1]{Gokhan Alcan}
\ead{gokhan.alcan@tuni.fi}

\author[1]{Jouni Mattila}
\ead{jouni.mattila@tuni.fi}

\affiliation[1]{organization={Department of Engineering and Natural Sciences, Tampere University}, city={Tampere}, country={Finland}} \cortext[1]{Corresponding author: Mahdi Hejrati}

\begin{abstract}
This paper develops an adaptive modular geometric control framework for robotic manipulators with uncertain inertial parameters. The manipulator is decomposed into rigid-body and joint modules, where each rigid-body module is represented by an Euler--Poincaré-type spatial dynamics on the Lie algebra \(\mathfrak{se}(3)\), and configuration errors are defined intrinsically through logarithmic maps on \(SE(3)\). The joint modules impose local screw constraints that relate adjacent body twists, accelerations, and transmitted wrenches, yielding a recursive propagation structure for the interconnected multibody system. Within this formulation, local geometric control laws are constructed at the module level, while the interconnection among modules is characterized by power-conjugate twist--wrench pairs induced by the natural duality pairing between the Lie algebra \(\mathfrak{se}(3)\) and its dual space \(\mathfrak{se}(3)^*\). For the nominal case, exponential tracking stability of the interconnected system is established using local configuration energy functions on \(SE(3)\) and the power-preserving structure of the modular interconnection. To address inertial parametric uncertainty, a geometric adaptation law is introduced on the manifold of symmetric positive-definite matrices, ensuring physically consistent parameter estimates while retaining compatibility with the Lie-algebraic control formulation. Under the adaptive controller, semi-global uniform ultimate boundedness of the closed-loop tracking and parameter estimation errors is proven. Numerical simulations on a redundant high-inertia robotic manipulator demonstrate accurate pose tracking, smooth transient behavior, orientation regulation, and robustness under inertial perturbations. Comparative studies with state-of-the-art methods further illustrate the effectiveness of the proposed framework for complex robotic manipulation tasks.
\end{abstract}


\maketitle

\section{INTRODUCTION}

Robotic manipulators constitute a fundamental class of interconnected multi-body systems, composed of rigid bodies coupled through kinematic joints. Owing to their structural versatility and ability to generate controlled motion and force in complex environments, they are widely employed in industrial automation, heavy-duty machinery, hydraulic excavators, cooperative and parallel robots, rehabilitation and assistive exoskeletons, healthcare technologies, logistics, and field robotics. In many of these applications, manipulators are required to operate under demanding conditions involving large inertial loads, strong dynamic coupling, redundancy, uncertainty, and stringent accuracy requirements. Consequently, the development of accurate, robust, and computationally efficient motion control strategies for robotic manipulators remains a central problem in robotics and control engineering.

A large body of literature has addressed the motion control of robotic manipulators through sophisticated methodologies, including adaptive control~\cite{slotine1987adaptive}, backstepping control~\cite{krstic1995nonlinear}, reinforcement-learning-based control~\cite{guo2026adaptive}, and model predictive control~\cite{huang2026switching}. Many of these formulations are expressed in joint coordinates, task-space coordinates, or other local Euclidean representations. While such approaches have been highly successful, robotic manipulators are not merely collections of joint variables; they are mechanical systems whose motion evolves on structured configuration spaces. In particular, the motion of each rigid body naturally belongs to the Special Euclidean group \(SE(3)\), where translational and rotational motions are coupled through the geometry of rigid-body displacement. Differential geometry therefore provides a natural language for modeling and controlling such systems, since it allows configuration errors, velocities, forces, and interconnections to be expressed in a coordinate-invariant manner \cite{bullo2005geometric}. This viewpoint is especially attractive for complex manipulators, where local coordinates or separate position--orientation representations may obscure the intrinsic structure of rigid-body motion and introduce representation-dependent complications.

Motivated by this perspective, geometric control has emerged as a powerful framework for systems evolving on Lie groups and differentiable manifolds. Early contributions established intrinsic control formulations for mechanical systems, including PD control on the Euclidean group~\cite{bullo_murray_1995}, tracking control for fully actuated mechanical systems~\cite{bullo1999tracking}, and systematic differential-geometric treatments of mechanical control systems~\cite{lewis1995aspects}. Building on these foundations, recent studies have developed geometric control designs directly on \(SE(3)\), with applications including quadrotor control~\cite{lee2010geometric,zhou2026semi} and geometric impedance control for robotic manipulators~\cite{seo2023geometric}. A key advantage of the geometric framework lies in its unified treatment of position and orientation errors, in contrast to conventional formulations that handle them separately. This intrinsic representation has been shown to improve tracking behavior and to enhance the success rate of learning-based manipulation methods~\cite{seo2023contact}. Beyond feedback control, geometric principles have also been employed in dynamic parameter identification~\cite{lee2018geometric,lee2019geometric}, physically consistent adaptive control~\cite{lee2018natural}, trajectory optimization on manifolds~\cite{alcan2025constrained}, equivariant robot learning~\cite{seo2025se}, geometry-aware movement primitives~\cite{abu2020geometry}, and stable skill learning on Riemannian manifolds~\cite{saveriano2023learning}. Collectively, these results demonstrate that differential geometric formulations provide both conceptual clarity and practical advantages for the analysis, control, and learning of robotic manipulators \cite{murray1997nonlinear}.

In parallel with the geometric viewpoint, modularity has become an important design principle for complex robotic systems. Rather than treating the manipulator as a monolithic dynamical system, modular control seeks to exploit its physical decomposition into interacting subsystems. This philosophy enables subsystem-level design, reduced computational burden, structural flexibility, and systematic stability analysis~\cite{nah2025modular}. Within the broader class of modular approaches, port-Hamiltonian modeling provide an important framework for representing complex physical systems in terms of energy storage, dissipation, and power-preserving interconnections~\cite{van2014port}. In this formulation, the interconnections of subsystems are described through energy variables and conjugate effort--flow pairs in a way that the exchanged power is explicitly accounted for at the system level. This makes the port-Hamiltonian framework particularly suitable for mechanical systems, where stability is naturally linked to energy balance and passivity properties. Accordingly, port-Hamiltonian-based control approaches typically exploit passivity as the main stability mechanism, using the Hamiltonian or a shaped energy function as the storage function in the closed-loop analysis~\cite{kumar2023tracking}. This viewpoint has led to structured formulations for rigid-body systems~\cite{rashad2022energy} and chains of interconnected rigid bodies~\cite{dirksz2011port}, where the modular interconnection structure is explicitly retained. These developments show that, for complex robotic systems, preserving the physical interconnection structure is not merely a modeling preference, but a control-theoretic mechanism through which stability can be systematically established.

In a closely related spirit to port-Hamiltonian approach, virtual decomposition control (VDC) ~\cite{zhu2010virtual} provides a modular control framework tailored to robotic manipulators and complex multi-body systems. Instead of interconnecting subsystems through Hamiltonian effort--flow ports, VDC decomposes the robot into dynamically consistent modules coupled through virtual power flows. These virtual power flows represent the power exchanged among adjacent modules and cancel when the decomposed subsystems are recomposed into the original robot, thereby providing a power-preserving mechanism for system-level stability~\cite{hejrati2025desired}. VDC has demonstrated high-performance control in modular manipulators~\cite{zhu2011virtual}, heavy-duty hydraulic systems~\cite{hejrati2025impact}, electric and lightweight manipulators~\cite{barjini2025surrogate,ding2022vdc}, and immersive bilateral teleoperation~\cite{hejrati2025robust}, with further theoretical extensions to strict-feedback nonlinear systems~\cite{koivumaki2022subsystem}.

More recently, modularity has also been revisited from a motor-primitives perspective, inspired by biological motor control~\cite{nah2025modular}. In this view, complex robotic behaviors are constructed from a small set of elementary motion and interaction modules, such as virtual trajectories and mechanical impedances~\cite{nah2025modular}. A central objective of this paradigm is to ensure that modules can be independently specified and combined while preserving closed-loop stability, particularly in contact-rich tasks. By exploiting the superposition of virtual trajectories and mechanical impedances, such approaches provide a constructive mechanism for composing task-space position, orientation, and joint-space behaviors without relying on inverse kinematics, while passivity is used to maintain stable physical interaction. This recent direction further reinforces the importance of modularity as a design principle for robotic control, especially when complex behaviors must be generated from simpler, physically interpretable building blocks.

Although the aforementioned studies have achieved significant progress in robotic manipulator control, several aspects remain to be further integrated within a unified framework. First, many existing formulations are developed in joint coordinates or task-space coordinates. While these descriptions are effective and widely used, they do not always reflect the intrinsic differential-geometric nature of rigid-body motion, where each body evolves on a nonlinear configuration space. Second, modularity is often introduced at the level of large interacting subsystems, for example by considering the manipulator as one module coupled to an environment, an object, or another robotic system. This viewpoint is valuable for system-level composition, but it does not fully exploit the internal modular structure of the manipulator itself. A further issue concerns the relation between geometric adaptation and control design. In many cases, the adaptation mechanism is constructed to preserve the physical consistency of the estimated inertial parameters, whereas the tracking errors and control actions are formulated in the Euclidean space. As a result, the adaptation law and the feedback control law may not be expressed within the underlying natural geometric structure. Moreover, the stability guarantees of modular and energy-based formulations are frequently established through passivity, \(L_2/L_\infty\) boundedness, or related input--output arguments. These tools are powerful and physically meaningful; however, for high-performance tracking, it is also desirable to obtain explicit exponential convergence guarantees, at least in the nominal case where unknown uncertainties are absent.

These observations motivate the development of a unified adaptive modular geometric control framework. The objective is to formulate both the control law and the inertial adaptation law in geometrically consistent spaces, while extending the notion of modularity from the system level to the internal structure of the manipulator. In this framework, rigid bodies and joints are treated as the constitutive modules of the robot, and the same modular viewpoint can be extended to additional interacting components, such as external objects, environments, or other robots. Such a formulation allows the intrinsic geometry of rigid-body motion, the interconnection structure of the manipulator, the power exchanged among modules, and the adaptation of uncertain inertial parameters to be handled coherently, while ensuring exponential tracking convergence in the absence of unknown uncertainties.

The main contributions of this paper are summarized as follows.
\begin{itemize}
\item A modular geometric control formulation is developed for robotic manipulators by decomposing the system into rigid-body and joint modules. Each rigid-body module is modeled using an Euler--Poincaré-type spatial representation on the Lie algebra, while logarithmic configuration errors are defined intrinsically on \(SE(3)\). Joint screw axes impose the local kinematic constraints, and recursive propagation of body twists, accelerations, and wrenches ensures consistency across the interconnected modules. This yields a link-level geometric control structure that preserves the intrinsic rigid-body geometry while remaining compatible with the physical interconnection of the manipulator.
\item Exponential tracking stability of the complete interconnected manipulator is established for the proposed modular geometric controller. The proof exploits local configuration energy functions on \(SE(3)\), body velocity errors, joint-module dynamics, and virtual power flows associated with the decomposed subsystems. The joint screw coordinates provide both the kinematic constraints between adjacent rigid bodies and the power-conjugate relation between joint velocities and transmitted wrenches, allowing the module-level stability properties to be lifted to the full manipulator.
\item The framework is extended to manipulators with uncertain inertial parameters through a geometric adaptation law defined on the manifold of symmetric positive-definite matrices. The proposed adaptation mechanism preserves physical consistency of the estimated inertial parameters and avoids drift from the admissible parameter set. By incorporating this law into the modular geometric controller, the adaptive closed-loop system retains the underlying natural geometric structure, and semi-global uniform ultimate boundedness is established under parametric uncertainty.
\item The proposed control framework is implemented and evaluated on a redundant high-inertia robotic manipulator. The simulation study examines tracking accuracy, transient response, orientation regulation, control effort, and robustness with respect to inertial uncertainty. Comparisons with representative control approaches are provided to illustrate the practical relevance of the theoretical development in a challenging manipulation scenario.
\end{itemize}

The remainder of the paper is organized as follows. Section \ref{Mathematical Preliminaries} reviews the geometric preliminaries on \(SE(3)\) and the modular decomposition framework. Section \ref{Modular Geometric Control} develops the modular geometric controller, and Section \ref{Adaptive Modular Geometric Control} introduces the adaptive extension. Section \ref{Results} presents simulation results, followed by concluding remarks in Section \ref{Conclusion}.

\section{Mathematical Preliminaries}
\label{Mathematical Preliminaries}
\subsection{Notation and Geometric Preliminaries}

We denote by $\mathbb{R}^n$ the $n$-dimensional Euclidean space with the standard inner product.
The special orthogonal (SO(3)) group and special Euclidean \(SE(3)\) group are defined as,
\begin{equation}
    SO(3) := \{ R \in \mathbb{R}^{3\times 3} \mid R^\top R = I,\ \det(R)=1\},
\end{equation}
\begin{equation}
    SE(3) := \left\{
        \mathcal{T} =
        \begin{bmatrix}
            R & p\\
            0 & 1
        \end{bmatrix}
        \,\middle|\,
        R \in SO(3),\ p \in \mathbb{R}^3
    \right\}.
\end{equation}
with \((\cdot)^\top\) being the transpose operator. The Lie algebra of $SE(3)$ is denoted by $\mathfrak{se}(3)$ and is given by,
\begin{equation}
    \mathfrak{se}(3) := \left\{
        \hat{\xi} =
        \begin{bmatrix}
            \hat{\omega} & v\\
            0 & 0
        \end{bmatrix}
        \,\middle|\,
        \omega,v \in \mathbb{R}^3
    \right\},
\end{equation}
where $(\cdot)^\wedge : \mathbb{R}^6 \to \mathfrak{se}(3)$ is the standard skew-symmetric “hat” operator with its inverse map denoted by $(\cdot)^\vee$ so that $\omega = \hat{\omega}^\vee$. The Lie bracket of two elements \(\hat{V}_1 \in \mathfrak{se}(3)\) and \(\hat{V}_2 \in \mathfrak{se}(3)\) is defined as \([\hat{V}_1, \hat{V}_2] = \hat{V}_1\,\hat{V}_2-\hat{V}_2\,\hat{V}_1\). The dual space of \(\mathfrak{se}(3)\), denoted as \(\mathfrak{se}^*(3)\), is the space of the wrenches applied to a rigid body with the pair of \((\tau,f)\). The Lie group exponential map, $\exp : \mathfrak{se}(3) \to SE(3)$,  provides a smooth mapping from the Lie algebra to the group, with local inverse given by the logarithm map, $\log : SE(3) \to \mathfrak{se}(3)$.

We use $\mathscr{P}(4)$ to denote the manifold of symmetric positive definite $4\times 4$ matrices,
\begin{equation}
    \mathscr{P}(4) := \{ L \in \mathbb{R}^{4\times 4} \mid L^\top = L,\ L \succ 0\}.
\end{equation}
Since $\mathscr{P}(4)$ is an open subset of the vector space of symmetric matrices,
the tangent space at each $L \in \mathscr{P}(4)$ can be identified with
$T_L\mathscr{P}(4)\cong \mathrm{Sym}(4) := \{ X \in \mathbb{R}^{4\times 4} \mid X^\top = X\}$.

\begin{defka}\label{def:metrics}
Let \(\mathbb{G}_{(\mathscr{Q},\mathscr{W})}\) denote an inner product defined on a space \(\mathscr{Q}\), with \(\mathscr{W}\succ 0\) representing either a weighting matrix or, when appropriate, a base point on \(\mathscr{Q}\). In particular, we consider the following cases:
\begin{equation}
    \mathbb{G}_{(\mathbb{R}^n,W)}(x,y)
    =
    x^\top W y,
\end{equation}
\begin{equation}\label{eq:inner_product_se3}
    \mathbb{G}_{(\mathfrak{se}(3),K)}(\hat{V}_1,\hat{V}_2)
    =
    V_1^\top K V_2,
\end{equation}
\begin{equation}
    \mathbb{G}_{(T_L\mathscr{P}(4),L)}(X,Y)
    =
    \frac{1}{2}\,\mathrm{tr}\!\left(L^{-1} X L^{-1} Y\right),
    \label{Metric}
\end{equation}
for \(x,y\in\mathbb{R}^n\), \(\hat{V}_1,\hat{V}_2\in\mathfrak{se}(3)\), and \(X,Y\in T_L\mathscr{P}(4)\). Then, the associated norm can be defined by,
\begin{equation}
    \|z\|_{(\mathscr{Q},\mathscr{W})}^2
    :=
    \mathbb{G}_{(\mathscr{Q},\mathscr{W})}(z,z),
\end{equation}
where the subscript indicates the space and weighting (or base point) with respect to which the norm is defined.
\end{defka}

The affine-invariant Riemannian metric on \(\mathscr{P}(4)\) is invariant under the congruence action
$\mathcal{A}*L = \mathcal{A}L\mathcal{A}^\top$,
for $\mathcal{A}\in GL(4)$ being a non-singular matrix. Accordingly, for \(X,Y\in T_L\mathscr{P}(4)\), the tangent vectors are transformed as
$
X \mapsto \mathcal{A}X\mathcal{A}^\top,
$
$
Y \mapsto \mathcal{A}Y\mathcal{A}^\top,
$
and the metric satisfies $
\mathbb{G}_{(T_L\mathscr{P}(4),\mathcal{A}L\mathcal{A}^\top)}
\left(\mathcal{A}X\mathcal{A}^\top,\mathcal{A}Y\mathcal{A}^\top\right)
=
\mathbb{G}_{(T_L\mathscr{P}(4),L)}(X,Y).
$

Additionally, an inner product
\(\mathbb{G}_{(\mathfrak{se}(3),K)}\) defined on the Lie algebra
\(\mathfrak{se}(3)\) induces a left-invariant Riemannian metric on
\(SE(3)\). Let \(l_{\mathscr{T}}:SE(3)\rightarrow SE(3)\) denote left
multiplication by \(\mathscr{T}\in SE(3)\), and let
\((l_{\mathscr{T}})_*\) be its pushforward. For tangent vectors
\(\mathscr{T}\hat V_1,\mathscr{T}\hat V_2\in T_{\mathscr{T}}SE(3)\),
with \(\hat V_1,\hat V_2\in\mathfrak{se}(3)\), the left-invariant metric is
defined by pulling them back to the identity as
\begin{equation}
\begin{split}
    \mathbb{G}_{(T_{\mathscr{T}}SE(3),K)}
    \left(\mathscr{T}\hat V_1,\mathscr{T}\hat V_2\right)
    &=
    \mathbb{G}_{(\mathfrak{se}(3),K)}
    \left((l_{\mathscr{T}^{-1}})_*(\mathscr{T}\hat V_1),
          (l_{\mathscr{T}^{-1}})_*(\mathscr{T}\hat V_2)\right)  \\
    &=
    \mathbb{G}_{(\mathfrak{se}(3),K)}
    \left(\hat V_1,\hat V_2\right).
\end{split}
\label{left_invar}
\end{equation}

\subsection{Manipulator Kinematics}
For the \(n\)-DoF interconnected rigid bodies shown in Fig. \ref{fig:interconnected}, \(\mathcal{P}\in SE(3)\) and \(\mathcal{P}^d\in SE(3)\) represent the end-effector actual and desired pose. Such system can be decomposed into rigid body and joint subsystems, as shown in Fig. \ref{fig:interconnected}.
\newpage
\begin{assumption}
For every admissible actual and desired profile in \(\mathcal P\) and \(\mathcal P^d\), there exist corresponding collections of local body transformations \(\mathcal T=[\mathcal T_i]\) and \(\mathcal T^d=[\mathcal T_i^d]\), with \(\mathcal T_i,\mathcal T_i^d\in SE(3)\), that are consistent with the underlying kinematic constraints of the decomposed system. Moreover, on the domain of interest, such correspondences are assumed to be represented by locally well-defined mappings \(\mathscr K:\mathcal P\to\mathcal T\) and \(\mathscr K^d:\mathcal P^d\to\mathcal T^d\).
\end{assumption}
\begin{rmk}
Closed-form expressions for the mappings $\mathscr K$ and $\mathscr K^d$ are not required. They may be obtained analytically when possible, or numerically through inverse-kinematic or optimization-based procedures. In the present development, only the existence of a consistent realization on the domain of interest is required.
\end{rmk}

We desire to derive the local homogeneous matrices with respect to the preceding bodies as,
\begin{equation}
    \begin{split}
        \mathcal{T}_{i-1,i}(\theta_i) = \mathcal{T}_{i-1,i}(0)e^{\hat{\xi}_i\theta_i},
    \end{split}
    \label{Ti}
\end{equation}
where \(\hat{\xi}_i \in \mathfrak{se}(3)\) is the constant twist of the body frame and \(\mathcal{T}_{i-1,i}\) indicates the transformation matrix of frame \{i\} with respect to \{i-1\}, with \(\mathcal{T}_{i-1,i}(0)\) being the initial configuration between the frames. Additionally, \(\theta \in \mathbb{S}^n\) and \(\dot{\theta} \in T_\theta\mathbb{S}^n\) with \(\mathbb{S}^n = \mathbb{S}^1 \times...\times \mathbb{S}^1\) is the joint generalized coordinate. We can derive \(\mathcal{T}^d_i\) in the sense of (\ref{Ti}) as,
\begin{equation}
    \begin{split}
        \mathcal{T}^d_{i-1,i}(\theta_i^d) = \mathcal{T}_{i-1,i}(0) e^{\hat{\xi}_i\theta^d_i}.
    \end{split}
    \label{Tid}
\end{equation}
In this formulation, the motion of the individual rigid bodies is preserved on \(SE(3)\) rather than being reduced to joint-space coordinates. The local transformation matrices expressed in the inertial frame can be written as \(\mathcal{T}_j = \prod_{i=1}^{j} \mathcal{T}_{i-1,i}(\theta_i)\), where the product follows the recursive chain structure of the mechanism.

Let $\mathcal{T}_i, \mathcal{T}_i^d \in \mathrm{SE}(3)$ denote the actual and desired poses of rigid body $i$, then the local configuration error can be defined as,
\begin{equation}
    e_i = (\mathcal{T}^d_{i})^{-1}\,\mathcal{T}_i = 
    \begin{bmatrix}
        (R^d_{i})^\top R_i & (R^d_{i})^\top (p_i - p^d_{i}) \\
        0 & 1
    \end{bmatrix}.
    \label{e}
\end{equation}
The associated Lie-algebraic local configuration error is represented as,
\begin{equation}
\eta_i := \log(e_i)^\vee \in \mathbb{R}^6.
\label{eq:log_error}
\end{equation}
\begin{figure*}
    \centering
    \includegraphics[width=0.9\linewidth]{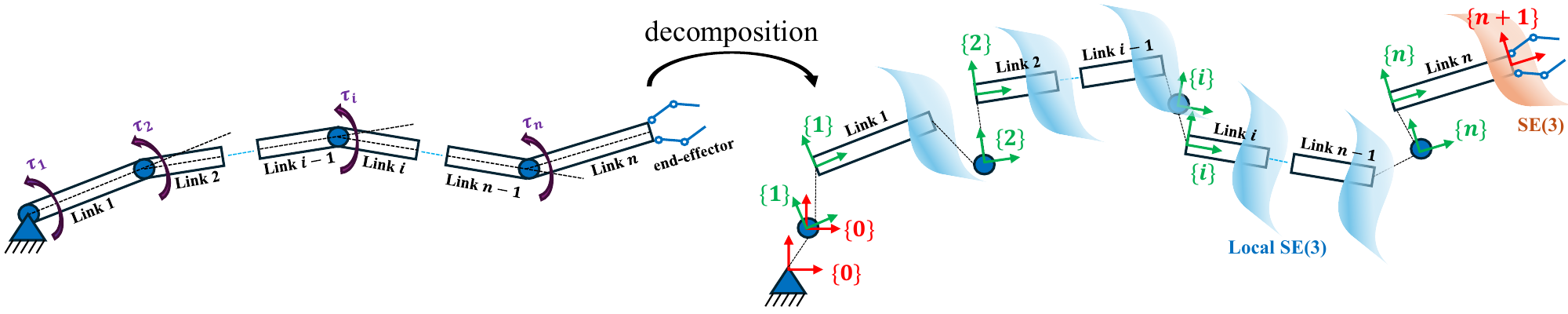}
    \caption{Interconnected multi-rigid body system decomposed into modules of rigid body and joints.}
    \label{fig:interconnected}
\end{figure*}
\begin{defka}
\label{Def log}
    Given $K_{z_i} \succ 0$, the local configuration energy error of body $i$ is defined in the sense of (\ref{eq:inner_product_se3}) as,
\begin{equation}
\Psi_i(e_i) = \frac{1}{2} \mathbb{G}_{(\mathfrak{se}(3),K_{zi})}(\hat{\eta},\hat \eta) = \frac{1}{2} \eta_i^\top K_{z_i} \eta_i .
\label{eq:configuration_energy}
\end{equation}
Since $K_{z_i}$ is positive definite, $\Psi_i$ is locally positive definite with respect to the identity $e_i = I$ and provides a quadratic measure of configuration deviation within the injectivity domain of the logarithm map.
\end{defka}

For each decomposed rigid body floating in space, we can write,
\begin{equation}
    \Dot{\mathcal{T}}_{i} = \mathcal{T}_{i} \hat{\mathscr{V}}_i,
    \label{Vi_O}
\end{equation}
with \(\hat{\mathscr{V}}_i = (\hat{w}_i, v_i) \in \mathfrak{se}(3)\) is the body-fixed velocity expressed in the body frame with respect to the spatial frame and \(\Dot{\mathcal{T}}_{i} \in T_{{\mathcal{T}}_{i}}SE(3)\) is the tangent vector of \({\mathcal{T}}_{i}\). In order to keep the modularity of the control framework through recursive approach, we extend (\ref{Vi_O}) to,
\begin{equation}
    {\mathscr{V}}_i = Ad_{\mathcal{T}^{-1}_{i-1,i}}\,{\mathscr{V}}_{i-1} + {\mathscr{V}}_{i-1,i},
\end{equation}
where by using \(\hat{\mathscr{V}}_{i-1,i} = \mathcal{T}^{-1}_{i-1,i} \Dot{\mathcal{T}}_{i-1,i}\) and replacing from (\ref{Ti}), one can establish,
\begin{equation}
    {\mathscr{V}}_i = Ad_{\mathcal{T}^{-1}_{i-1,i}}\,{\mathscr{V}}_{i-1} + \xi_i\, \Dot{\theta}_i,
    \label{Vi}
\end{equation}
where \(Ad_\mathcal{T}\) is the Adjoint map. As a result, (\ref{Vi}) represents the body-fixed velocity of each local rigid body in terms of adjacent bodies. By taking the time-derivative of (\ref{Vi}), one can establish the analytic representation of body-fixed acceleration,
\begin{equation}
    \mathscr{A}_i = Ad_{\mathcal{T}^{-1}_{i-1,i}}\,\mathscr{A}_{i-1}-\,[\xi_i\, \Dot{\theta}_i,Ad_{\mathcal{T}^{-1}_{i-1,i}}\mathscr{V}_{i-1}] + \xi_i\, \Ddot{\theta}_i.
    \label{Ai}
\end{equation}
The desired terms $\mathscr{V}^d_i$ and $\mathscr{A}^d_i$ are computed by replacing $\theta_i$ and $\dot\theta_i$ with $\theta_i^d$ and $\dot\theta_i^d$ in \eqref{Vi} and \eqref{Ai}, respectively.

\subsection{Manipulator Dynamics}
Given the geometric representation of the modular system kinematics, the dynamics of each local subsystem can be expressed in terms of spatial wrenches. Consider a single decomposed rigid body in space. The net wrench acting on body $i$ is given by,
\begin{equation}
    \mathscr{F}_i = F_i - Ad_{\mathcal{T}^{-1}_{i,i+1}}^{\top} F_{i+1},
\end{equation}
where $F_i \in \mathfrak{se}^*(3)$ and $F_{i+1} \in \mathfrak{se}^*(3)$ denote the body-frame representations of the wrenches applied at the origins of frames $\{i\}$ and $\{i+1\}$, respectively. On the other hand, the net wrench of rigid body $i$ can be written using the Euler--Poincaré equation as \cite{marsden1999introduction},
\begin{equation}
    \mathscr{F}_i = M_i \mathscr{A}_i - ad_{\mathscr{V}_i}^\top \left( M_i \mathscr{V}_i \right),
\end{equation}
where $M_i \in \mathbb{R}^{6\times6}$ is the spatial inertia matrix and $ad_{\mathscr{V}_i}^\top$ denotes the co-adjoint operator. Consequently, the dynamics of the rigid body module can be represented as,
\begin{equation}
    F_i = M_i \mathscr{A}_i - ad_{\mathscr{V}_i}^\top \left( M_i \mathscr{V}_i \right) + Ad_{\mathcal{T}^{-1}_{i,i+1}}^{\top} F_{i+1}.
    \label{Fi}
\end{equation}

\subsection{Stability Criterion of the Modular system}

So far, the modular geometric representation of the kinematics and dynamics has been derived for the decomposed interconnected multi-body system, which is consistent with the spatial rigid-body formulations in~\cite{murray2017mathematical, featherstone2008rigid, park1995lie}. In this section, we introduce the stability criterion used for each module. The main idea is that the interconnection between adjacent modules is represented through power-conjugate twist--wrench variables, so that the internal power exchanges cancel when the decomposed modules are recomposed into the complete system.

The virtual power flow can be motivated from the principle of virtual work. For a rigid body evolving on \(SE(3)\), an admissible virtual displacement at configuration \(\mathscr{T}_i\in SE(3)\) can be written as
$$
\delta \mathscr{T}_i
=
\mathscr{T}_i{\delta x}_i^{\wedge},
\qquad
\delta x_i\in \mathbb{R}^6\simeq\mathfrak{se}(3),
$$
where \(\delta x_i\) denotes the left-trivialized virtual displacement. A wrench \(F_i\in\mathfrak{se}^*(3)\) acts as a one-form on this virtual displacement, and the corresponding virtual work is given by the natural duality pairing
$$
\delta W_i
=
\langle F_i,\delta x_i\rangle ,
$$
where \(\langle\cdot,\cdot\rangle\) denotes the duality pairing between \(\mathfrak{se}^*(3)\) and \(\mathfrak{se}(3)\). Along a motion, the infinitesimal displacement is generated by the body twist \(\mathscr{V}_i\in\mathbb{R}^6\simeq\mathfrak{se}(3)\), so that \(\delta x_i=\mathscr{V}_i\,dt\). Hence, the mechanical work performed over a time interval is
$$
W_i
=
\int_{t_0}^{t_1}
\langle F_i,\mathscr{V}_i\rangle\,dt,
$$
and the integrand \(\langle F_i,\mathscr{V}_i\rangle\) represents the instantaneous mechanical power.

For a module in the decomposed multi-body system, let \((F_i^r,\mathscr{V}_i^r)\) denote the required wrench and body twist of the \(i\)-th module, and let \((F_i,\mathscr{V}_i)\) denote the corresponding actual quantities. The deviation of the actual module from its required behavior is characterized by the virtual effort and virtual flow variables
$$
\widetilde F_i = F_i^r-F_i,
\qquad
\widetilde{\mathscr{V}}_i
=
\mathscr{V}_i^r-\mathscr{V}_i .
$$
Here, \(\widetilde F_i\) represents the wrench mismatch associated with enforcing the required module behavior, whereas \(\widetilde{\mathscr{V}}_i\,dt\) represents the infinitesimal virtual displacement associated with the velocity mismatch. Therefore, by the same virtual-work principle, the differential virtual work associated with this error port is
$$
dW_i^{v}
=
\left\langle
\widetilde F_i,
\widetilde{\mathscr{V}}_i
\right\rangle dt .
$$
Consequently, the corresponding instantaneous rate of \textit{virtual} work is
\begin{equation}
p_i
=
\frac{dW_i^{v}}{dt}
=
\left\langle
\widetilde F_i,
\widetilde{\mathscr{V}}_i
\right\rangle
=
\left\langle
F_i^r-F_i,
\mathscr{V}_i^r-\mathscr{V}_i
\right\rangle .
\label{eq:vpf}
\end{equation}
This quantity is referred to as the \textit{virtual} power flow of the \(i\)-th module. It is not the physical power applied to the rigid body, which is \(\langle F_i,\mathscr{V}_i\rangle\), but rather the work rate associated with the virtual effort--flow pair measuring the deviation between the required and actual module behavior. When adjacent modules are interconnected, the internal virtual power flows enter the corresponding module inequalities with opposite signs. Therefore, in a properly composed modular system, these internal terms cancel when the module-level stability inequalities are summed, leaving only boundary power flows, dissipative terms, and uncertainty-dependent residual terms.

\begin{defka}[Dual map]
Let \(U\) and \(V\) be finite-dimensional vector spaces, and let
\(A:U\rightarrow V\) be a linear map. The dual map of \(A\), denoted by
\(A^*:V^*\rightarrow U^*\), is defined by
\begin{equation}
    \left\langle A^*\alpha,u\right\rangle
    =
    \left\langle \alpha,Au\right\rangle,
    \qquad
    \alpha\in V^*,\quad u\in U,
\end{equation}
where \(\langle\cdot,\cdot\rangle\) denotes the natural duality pairing.
\label{def:dual_map}
\end{defka}
In the present paper, the convention in Definition \ref{def:dual_map} is mainly applied to the previously defined maps \(\operatorname{Ad}_{\mathscr{T}}\) and \(\operatorname{ad}_{\mathscr{V}}\). Since twists and wrenches are represented by vectors in \(\mathbb{R}^6\) and paired through the standard power pairing, the corresponding dual maps are represented by matrix transposes, i.e.,
\begin{equation}
    \operatorname{Ad}_{\mathscr{T}}^*
    =
    \operatorname{Ad}_{\mathscr{T}}^\top,
    \qquad
    \operatorname{ad}_{\mathscr{V}}^*
    =
    \operatorname{ad}_{\mathscr{V}}^\top .
\end{equation}
This notation allows linear transformations to be transferred between twist and wrench variables while preserving the associated power pairing.

\begin{defka}
Let \(\mathrm{V}_i:\mathbb{R}^{n_i}\rightarrow\mathbb{R}_{\geq 0}\) be a continuously differentiable function satisfying
\begin{equation}
\underline{\alpha}_i\|x_i\|^2
\leq
\mathrm{V}_i(x_i)
\leq
\overline{\alpha}_i\|x_i\|^2,
\end{equation}
for all \(x_i\in\mathbb{R}^{n_i}\), where \(\underline{\alpha}_i,\overline{\alpha}_i>0\). Suppose that, along the trajectories of the \(i\)-th subsystem, the time derivative of \(\mathrm{V}_i\) satisfies
\begin{equation}
\dot{\mathrm{V}}_i(x_i,t)
\leq
-\widetilde{\alpha}_i\|x_i\|^2
+
p_i-p_{i+1}
+
\bar{\epsilon}_i,
\end{equation}
where \(\widetilde{\alpha}_i>0\), \(\bar{\epsilon}_i\geq 0\), and \(p_i,p_{i+1}\in\mathbb{R}\) denote the virtual power flows through the interconnection ports of the subsystem. Then, the \(i\)-th subsystem is said to be virtually stable.
\label{VS}
\end{defka}

\begin{rmk}
Definition~\ref{VS} shows that the stability of an individual module is expressed in terms of its local dissipation and the virtual power flows exchanged with neighboring modules. When the modules are interconnected, the internal virtual power flows cancel in the sum of the module-level inequalities. Hence, if \(\bar{\epsilon}_i=0\), the resulting inequality can be used to establish exponential stability of the interconnected nominal system. If \(\bar{\epsilon}_i>0\), the same structure leads to an ultimate boundedness result, with the residual bound determined by the uncertainty-dependent terms.
\end{rmk}

\section{Modular Geometric Control}
\label{Modular Geometric Control}
In this section we will elaborate on the design of modular geometric control method for an interconnected rigid body system in the absence of uncertainty. While the local configuration error \(\eta_i\) is defined in (\ref{eq:log_error}), extending this notion to the body velocity error in the geometric approach is less straightforward. From (\ref{Vi_O}) we can see that \(\Dot{\mathcal{T}}_{i} \in T_{{\mathcal{T}}_{i}}SE(3)\) while \(\Dot{\mathcal{T}}^d_{i} \in T_{{\mathcal{T}}^d_{i}}SE(3)\), which shows that actual and desired tangent vectors do not belong to the same tangent space. This implies that to compute the velocity error we cannot simply compare them as if they belong to the same space. In contrast, we will define the body velocity error as,
\begin{equation}
    \mathscr{V}^e_{i} = Ad_{e_i^{-1}}\mathscr{V}^d_{i} - \mathscr{V}_i,
    \label{Ve}
\end{equation}
where \(Ad_{e_i^{-1}}\) transports the desired velocity to the space where the \(\mathscr{V}_i\) resides. In the same sense, the required body velocity \(\mathscr{V}^r_{i} \in R^6\) can be defined as,
\begin{equation}
    \mathscr{V}^r_{i} = Ad_{e_i^{-1}}\mathscr{V}^d_{i} -\Gamma_i \eta_i,
    \label{Vri}
\end{equation}
with \(\Gamma_i\) being positive definite gain. Defining the required body-fixed velocity as a logarithmic function of local natural error allows smoother convergence to the origin \cite{bullo_murray_1995}. The definition of \(\mathscr{V}^r_{i}\) in the sense of (\ref{Vri}) enables the geometric motion of each single rigid bodies on the manifold instead of just considering the end-effector. The required body acceleration, on the other hand, can be computed by taking the time derivative of (\ref{Vri}) as,
\begin{equation}
    \mathscr{A}^r_{i} = Ad_{e_i^{-1}}\mathscr{A}^d_{i} +\,[\mathscr{V}^e_{i},Ad_{e_i^{-1}}\mathscr{V}^d_{i}] - \Gamma_i\, \dot \eta_i. 
\end{equation}
Now the required wrench of the \(i^{th}\) body can be expressed by replacing \(\mathscr{V}_{i}\) and \(\mathscr{A}_{i}\) in (\ref{Fi}) with \(\mathscr{V}^r_{i}\) and \(\mathscr{A}^r_{i}\), respectively,
\begin{equation}
    \begin{split}
        F^r_i = M_i\,\mathscr{A}^r_i- ad^\top_{\mathscr{V}_i}(M_i\, \mathscr{V}^r_i) &+ K_v(\mathscr{V}^r_{i}-\mathscr{V}_{i})+ Ad_{\mathcal{T}^{-1}_{i,i+1}}^\top\, F^r_{{i+1}} .
    \end{split}
    \label{Fri}
\end{equation}
The last term in~(\ref{Fri}) propagates the required wrench from the successive bodies. Finally, the control input applied at the joints is computed as,
\begin{equation}
    \tau^r_{i} = \xi_i^\top F^r_{i} + \mathscr{J}^r_{i},
    \label{tau_i}
\end{equation}
where $\mathscr{J}^r_i$ denotes the contribution arising from the joint dynamics. The required torque in~(\ref{tau_i}) results from the composition of the decomposed rigid-body and joint subsystems. The first term on the right-hand side accounts for the rigid-body motion on \(SE(3)\), where the screw axis $\xi_i$ enforces the joint constraint and maps the required wrench to the corresponding joint torque. The second term represents the contribution of the joint dynamics. A critical component in composing the required joint action is the required joint velocity $\dot{\theta}^r_i$, analogous to the required body velocity defined in~(\ref{Vri}). To derive this term, assume an electric actuator driving the $i$th joint with the following dynamics,
\begin{equation}
    I_{mi} \ddot{\theta}_i = \mathscr{J}_i,
    \label{Ji}
\end{equation}
where $\mathscr{J}_i$ denotes the net torque applied to the joint. The net torque is given by,
\begin{equation}
    \mathscr{J}_i = \tau_i - \xi_i^\top F_i,
    \label{Ji2}
\end{equation}
with $\tau_i$ being the actuator torque and $\xi_i^\top F_i$ representing the torque induced by the rigid-body wrench. The required joint action is, then, defined as,
\begin{equation}
    \mathscr{J}^r_{i} = I_{mi} \ddot{\theta}^r_{i} 
    + k_{ai}\bigl(\dot{\theta}^r_{i} - \dot{\theta}_{i}\bigr),
    \label{Jri}
\end{equation}
where $k_{ai} > 0$. The required joint velocity satisfies,
\begin{equation}
    \xi_i \dot{\theta}^r_{i}
    =
    \mathscr{V}^r_{i}
    -
    Ad_{\mathcal{T}^{-1}_{i-1,i}} \mathscr{V}^r_{i-1},
    \label{dthetar}
\end{equation}
with \(\xi_i^{-1} = \xi_i^\top\), which enforces kinematic consistency between adjacent subsystems.

\begin{rmk}
The joint-module control action in \eqref{Jri} also inherits a geometric character from the proposed rigid-body formulation. In particular, the required joint velocity in \eqref{dthetar} is not prescribed independently from a Euclidean joint-space tracking error. Instead, it is induced by the required body velocity of the adjacent rigid-body module through (\ref{Vri}), whose motion is defined intrinsically on \(SE(3)\). Therefore, the joint module enforces the kinematic constraint needed to realize the local geometric motion objective of the rigid body. This construction differs from conventional joint-space formulations, in which the desired rigid-body behavior is typically obtained as a consequence of joint-level error regulation in Euclidean coordinates.
\label{rmk:geometric_joint_action}
\end{rmk}

\begin{thm}
    Consider a complex robotic system which is decomposed into the subsystems, as shown in Fig. \ref{fig:interconnected}, with the rigid body and joint dynamics of (\ref{Fi}) and (\ref{Ji}), respectively. Each subsystem is virtually stable in the sense of Definition \ref{VS} under the modular geometric control (\ref{Fri}) and (\ref{Jri}).
\end{thm}
\begin{pf}
    We start the proof by considering the rigid-body subsystem. Let us define the Lyapunov function as,
    \begin{equation}
        \mathrm{V}_i = \|\mathscr{V}^r_{i}-\mathscr{V}_{i}\|^2_{(\mathbb{R}^6,M_i)} + \Psi_i(e_i),
        \label{v_i}
    \end{equation}
    with \(\Psi_i(e_i)\) defined in Definition \ref{Def log}. Subtracting (\ref{Fi}) from (\ref{Fri}), we have,
    \begin{equation}
    \begin{split}
        F^r_i- F_i &= M_i\, (\mathscr{A}^r_i-\mathscr{A}_i)- ad^\top_{\mathscr{V}_i}M_i\, (\mathscr{V}^r_i-\mathscr{V}_i)+K_v(\mathscr{V}^r_{i}-\mathscr{V}_{i})+ Ad_{\mathcal{T}^{-1}_{i,i+1}}^\top\, (F^r_{{i+1}}-F_{{i+1}}) .
    \end{split}
    \label{F_diff}
    \end{equation}
    Now, by taking the time derivative of (\ref{v_i}), using Definition \ref{def:dual_map}, and replacing from (\ref{F_diff}), we have,
    \begin{equation}
    \begin{split}
        \Dot{\mathrm{V}}_i
&= \underbrace{\left\langle
F_i^r-F_i,
\mathscr{V}_i^r-\mathscr{V}_i
\right\rangle}_{p_i}
-\|\mathscr{V}^r_{i}-\mathscr{V}_{i}\|^2_{(\mathbb{R}^6,K_v)}\quad
-\underbrace{\langle F^r_{{i+1}}-F_{{i+1}},Ad_{\mathcal{T}^{-1}_{i,i+1}}(\mathscr{V}^r_{i}-\mathscr{V}_{i})\rangle}_{p_{i+1}}\\
&\quad
+(\mathscr{V}^r_{i}-\mathscr{V}_{i})^\top\,ad^\top_{\mathscr{V}_i}M_i\, (\mathscr{V}^r_i-\mathscr{V}_i)+\dot\Psi_i(e_i)\\
        &=p_i-p_{i+1}-\|\mathscr{V}^r_{i}-\mathscr{V}_{i}\|^2_{(\mathbb{R}^6,K_v)}+\dot\Psi_i(e_i)+(\mathscr{V}^r_{i}-\mathscr{V}_{i})^\top\,ad^\top_{\mathscr{V}_i}M_i\, (\mathscr{V}^r_i-\mathscr{V}_i).
    \end{split}
    \label{dV1}
    \end{equation}
     To continue, we first try to derive \(\dot\Psi_i(e_i)\) by exploiting the variation principle. Consider \(\mathrm{\delta}\) as the variation operator and $\nabla$ as the gradient operator. Then, we can write, 
    \begin{equation}
        \begin{split}
            \mathrm{\delta}\Psi_i(e_i) = \mathbb{G}_{(T_\mathscr{T}\mathrm{SE}(3),K_{z_i})}(\nabla\Psi_i,\mathrm{\delta}e_i)
        \end{split}
    \end{equation}
    where by considering infinitesimal variation around \(e_i\), one can write \(\mathrm{\delta}e_i = e_i\, \delta\psi_i^{\wedge}\) with \(\delta\psi_i^\wedge \in \mathfrak{se}(3)\). Consequently, we can derive the following by exploiting the left-invariance property in (\ref{left_invar}) and \cite[Theorem 1]{teng2022lie} as,
    \begin{equation}
        \begin{split}
            \mathrm{\delta}\Psi_i(e_i) &= \mathbb{G}_{(T_{e_i}\mathrm{SE}(3),K_{z_i})}(e_i\,\hat \eta_i,e_i\,\mathrm{\delta}\psi_i^{\wedge})= \mathbb{G}_{(\mathfrak{se}(3),K_{z_i})}(\,\hat \eta_i,\,\mathrm{\delta}\psi_i^{\wedge})
        \end{split}
        \label{dPsi}
    \end{equation}
    From (\ref{dPsi}) and with slight abuse of notation along with $\delta\psi_i^{\wedge} = -\hat{\mathscr{V}}^e_{i} dt$, the time derivative \(\dot\Psi_i(e_i)\) can be computed as,
    \begin{equation}
        \begin{split}
            \Dot{\Psi}_i(e_i) = -\mathbb{G}_{(\mathfrak{se}(3),K_{z_i})}(\,\hat \eta_i,\,\hat{\mathscr{V}}^e_{i})
        \end{split}
        \label{dtPsi}
    \end{equation}
    with \(\dot e_i = -e_i\,\hat{\mathscr{V}}^e_{i}\).  
    
    On the other hand, in ~(\ref{dV1}), the matrix $ad^\top_{\mathscr{V}_i} M_i$ is not inherently skew-symmetric. To address this issue, we employ Young's inequality together with the upper bound assumption for \(\|ad_{\mathscr{V}_i}^\top M_i\| \le \mathscr{C}_i \|\mathscr{V}_i\|\) where $\mathscr{C}_i > 0$ is a constant depending on the inertia matrix. Accordingly, one can obtain,
    \begin{equation}
    \begin{split}
        (\mathscr{V}^r_{i}-\mathscr{V}_{i}&)^\top \,ad^\top_{\mathscr{V}_i}\,{M}_i\,(\mathscr{V}^r_{i}-\mathscr{V}_i)\leq(\mathscr{V}^r_{i}-\mathscr{V}_{i})^\top(\dfrac{1}{2\alpha_{i}}+2\alpha_{i}\,\Bar{\mathscr{C}}_i^2)(\mathscr{V}^r_{i}-\mathscr{V}_{i})
    \end{split}
    \label{ad_bound}
    \end{equation}
    with \(\sup_{t>0}\|\mathscr{V}_{i}\|_{(\mathbb{R}^6,I)}=\hat{\mathscr{C}}_i <\infty\) and \(\Bar{\mathscr{C}}_i = \mathscr{C}_i\,\hat{\mathscr{C}}_i\). Now, replacing boundedness assumption along with (\ref{dtPsi}) in (\ref{dV1}) yields to,
    \begin{equation}
        \begin{split}
            \Dot{\mathrm{V}}_i & \leq-\dfrac{1}{2}\|\mathscr{V}^r_{i}-\mathscr{V}_{i}\|^2_{(\mathbb{R}^6,\mu_i)}-\mathbb{G}_{(\mathfrak{se}(3),K_{z_i})}(\hat \eta_i,\hat{\mathscr{V}}^e_{i}) +p_i-p_{i+1},
        \end{split}
        \label{DVV2}
    \end{equation}
    with \(\mu_i = 2K_v-{1}/{\alpha_{i}}+\alpha_{i}\,\Bar{\mathscr{C}}_i^2\). Consequently, by using \(\mathscr{V}^r_{i}-\mathscr{V}_{i} = \mathscr{V}^e_{i} - \Gamma_i \eta_i\) through (\ref{Ve}) and (\ref{Vri}), employing Cauchy–Schwarz and Young's inequalities and \(\rho_i = \max{(\lambda_{max}(M_i)},\lambda_{max}(K_{z_i}))\), with \(\lambda_{max}(.)\) being maximum eigenvalue, one can rewrite (\ref{DVV2}) as,
    \begin{equation}
        \begin{split}
            \Dot{\mathrm{V}}_i & \leq -\varrho_i{\mathrm{V}}_i +p_i-p_{i+1},
        \end{split}
        \label{dv}
    \end{equation}
    with \(\varrho_i = {\min(\mu_i, 2\bar{\mu}_i)}/\rho_i/8\), \(\mu_i \succ 0 \), \(\bar{K}_{z_i} = \Gamma_i\mu_i-K_{z_i} \succ 0\), \(\bar{\mu}_i = \Gamma_i\mu_i\Gamma_i-2\epsilon^{-1}\bar{K}_{z_i} \succ 0\), and \(2\epsilon_i \leq \mu_i\bar{K}_{z_i}^{-1}\). It follows from (\ref{dv}) that the virtual exponential stability of the rigid body subsystem is achieved.

    In order to investigate the stability of the joint subsystems, we define,
    \begin{equation}
        \mathrm{V}_{ai} = \dfrac{1}{2}I_{mi}(\Dot{\theta}^r_i-\Dot{\theta}_i)^2,
        \label{Va}
    \end{equation}
    and by subtracting (\ref{Ji}) from (\ref{Jri}) and replacing in the time derivative of (\ref{Va}), we have,
    \begin{equation}
        \Dot{\mathrm{V}}_{ai} = (\Dot{\theta}^r_i-\Dot{\theta}_i)\,(\mathscr{J}^r_{i}-\mathscr{J}_{i})-\,k_{ai}(\Dot{\theta}^r_{i}-\Dot{\theta}_{i})^2.
        \label{DVa1}
    \end{equation}
    Now, by using (\ref{tau_i}) and (\ref{Ji2}), we have,
    \begin{equation}
        \begin{split}
            \Dot{\mathrm{V}}_{ai} = (\Dot{\theta}^r_i-\Dot{\theta}_i)\,\xi_i^\top\,(\,F^r_{i}-\,F_{i})-\,k_{ai}(\Dot{\theta}^r_{i}-\Dot{\theta}_{i})^2,
        \end{split}
        \label{Dva2}
    \end{equation}
    where we assumed \(\tau_i = \tau^r_i\) which is a reasonable assumption in the torque-controlled mode for the driver of the actuator. Now, we can extend (\ref{Dva2}) by using (\ref{dthetar}), Definition \ref{def:dual_map}, and (\ref{Vi}) as,
    \begin{equation}
        \begin{split}
            \Dot{\mathrm{V}}_{ai} &= -\left\langle
F_i^r-F_i,
\mathscr{V}_i^r-\mathscr{V}_i
\right\rangle-\,k_{ai}(\Dot{\theta}^r_{i}-\Dot{\theta}_{i})^2 - \left\langle  \,F^r_{i}-\,F_{i},Ad_{\mathcal{T}^{-1}_{i-1,i}}\,(\mathscr{V}^r_{{i-1}}-\mathscr{V}_{{i-1}})\right \rangle\\
            & = -p_i+p_{i-1}-\,\dfrac{2k_{ai}}{I_{mi}}\mathrm{V}_{ai},
        \end{split}
        \label{DVa}
    \end{equation}
   demonstrating the virtual exponential stability of subsystem.
\end{pf}

\begin{thm}
    The system of inter-connected rigid bodies decomposed into subsystems of rigid body and joint modules with the governing dynamics of (\ref{Fi}) and (\ref{Ji}) under the control action of (\ref{Fri}) and (\ref{Jri}) is exponentially stable.
\end{thm}
\begin{pf}
    We define the Lyapunov function of the unified system as,
    \begin{equation}
        \mathrm{V} = \sum_{i=1}^n(\mathrm{V}_i+{\mathrm{V}}_{ai}),
        \label{V}
    \end{equation}
    which by taking its time derivative and replacing from Theorem 1 and the fact that \(\sum_{i=1}^n-p_i+p_{i-1}+p_i-p_{i+1} = p_0-p_{n+1} = 0\) with the assumption of fixed ground for the base and no interaction with environment for the last frame, we have,
    \begin{equation}
        \Dot{\mathrm{V}} \leq -\underline{\varrho}{\mathrm{V}},
        \label{DV}
    \end{equation}
    with \(\underline{\varrho} = \min_{i=1...n}({\varrho}_i,{2k_{ai}}/{I_{mi}})\). Consequently, one can show the exponential convergence of \(V \leq e^{-\underline{\varrho}t}V(0) \leq e^{-\underline{\varrho}t}\overline{\varrho} \|x(0)\|^2_{(\mathbb{R}^6,I)}\) with \(\overline{\varrho} = \max(\rho_i, I_{m_i})\), which following from (\ref{DV}), (\ref{V}), (\ref{v_i}), and (\ref{eq:configuration_energy}) leads to \(\|\Psi(e_i)\|^2_{(\mathfrak{se}(3),K_{zi})} \leq e^{-\underline{\varrho}t} \overline{\varrho} \|x(0)\|^2_{(\mathbb{R}^6,I)}\), demonstrating the non-singularity of the logarithmic map with the proper initialization. Moreover, \(x_i = [\Dot{\theta}_i^r-\Dot{\theta}_i, (\mathscr{V}^r_i-\mathscr{V}_i)^\top,\eta_i^\top]_{i=1}^n\).
    
    Now, we will show that our assumption regarding upper bound of the velocity, \(\sup_{t>0}\|\mathscr{V}_{i}\|_{(\mathbb{R}^6,I)}=\hat{\mathscr{C}}_i\) is correct. From (\ref{DV}), it can be shown that \(\sup_{t>0}\|\mathscr{V}^r_{i}-\mathscr{V}_{i}\|_{(\mathbb{R}^6,I)}\leq \sqrt{2\overline{\varrho}/\lambda_{min}(M_i)}\|x(0)\|_{(\mathbb{R}^6,I)}\). Then, we can write by exploiting (\ref{Ve}) and (\ref{Vri}), 
    \begin{equation}
        \sup_{t>0}\|\mathscr{V}_{i}\| \leq \sup_{t>0} \|\mathscr{V}^r_{i}-\mathscr{V}_{i}\| + \sup_{t>0}\|Ad_{e_i^{-1}}\mathscr{V}^d_{i}\|+\sup_{t>0}\|\Gamma_i\eta_i\|.
        \label{supV1}
    \end{equation}
    The inequality in (\ref{supV1}) can be written in the following form by exploiting the first-order approximation of \(Ad_{e^{-1}_i} \leq 1 + \|\eta_i\|\), boundedness of the desired velocity \(\|\mathscr{V}^d_{i}\| \leq N^d_{i}\),
     \begin{equation}
    \begin{split}
        \sup_{t>0}||\mathscr{V}_{i}|| \leq &\left(\sqrt{\dfrac{2\overline{\varrho}}{\lambda_{min}(M_i)}}(1 + N^d_{i})\right. \left.+\sqrt{\dfrac{2\lambda_{max}(\Gamma_i)^2}{\lambda_{min}(K_{z_i})}}\right)\|x(0)\|.
    \end{split}
        \label{supV2}
    \end{equation}
\end{pf}
As a result, the exponential stability of the overall interconnected rigid-body system under the proposed controller is established. Furthermore, by proper initialization, the logarithmic error function is shown to remain nonsingular, and an appropriate upper bound on the actual velocity is introduced to account for the non-skew-symmetric structure of the dynamics.
\begin{rmk}
The condition \(\sup_{t\ge 0}\|V_i(t)\|<\infty\) should be interpreted as a bounded-operating-domain assumption. Such an assumption is natural for robotic manipulators, since actuator torque limits, velocity saturation, finite power capacity, and safety constraints prevent physically admissible motions from generating unbounded body velocities. In the present analysis, this bound is used to estimate the non-skew-symmetric term associated with the spatial rigid-body dynamics and does not impose any additional structural restriction on the proposed modular geometric controller. In fact, the inequality in \eqref{supV2} shows that this upper bound can be selected as large as required by the admissible operating range of the system, according to its actuation capacity and inertial characteristics.
\end{rmk}

\section{Adaptive Modular Geometric Control}
\label{Adaptive Modular Geometric Control}
The preceding development has established the modular geometric control framework for the nominal rigid-body system. We now extend this formulation to account for parametric uncertainty in the dynamics, which may otherwise compromise tracking performance. To this end, an adaptive modular geometric control scheme is developed in this section. Without loss of generality, the joint subsystems are assumed to retain the same control law as in the preceding section, and the adaptive design is therefore presented for the rigid-body modules.

Consider the equation of motion for a decomposed rigid body in space as in (\ref{Fi}). In order to tackle the parameter uncertainty in inertial matrix, we need to develop an adaptive control law as,
\begin{equation}
    F^r_{i} =  \hat{M}_i\,\mathscr{A}^r_{i} -ad^\top_{\mathscr{V}_i}\,(\hat{M}_i\,\mathscr{V}^r_{i})+K_v(\mathscr{V}^r_{i}-\mathscr{V}_{i})+ Ad_{\mathcal{T}^{-1}_{i,i+1}}^\top\, F^r_{i+1}.
    \label{Fr_ad}
\end{equation}
In the following, the geometric adaptation law is designed that updates the spatial inertia matrix.

Let \(L_i,\hat{L}_i \in \mathscr{P}(4)\) denote the true and estimated pseudo-inertia matrices of link \(i\), respectively. There is a one-to-one mapping between \(L_i\) and \(M_i\) as follow,
\begin{equation}
\begin{split}
    L_i = f(M_i) = f(I_{A_i},h_{A_i},m_i)=\begin{bmatrix}
        \Sigma_{A_i}& h_{{A_i}}\\
        h^\top_{A_i}& m_i
    \end{bmatrix},\\
    M_i = f^{-1}(\Sigma_{A_i},h_{A_i},m_i)=\begin{bmatrix}
        I_{A_i}& \hat h_{{A_i}}\\
        -\hat h_{A_i}& m_i\cdot I_3
    \end{bmatrix}
\end{split}
\end{equation}
with \(\Sigma_{A_i}=0.5\operatorname{tr}(I_{A_i})I_3-I_{A_i}\). Since \(\mathscr{P}(4)\) is the manifold of symmetric positive-definite matrices, the discrepancy between \(L_i\) and \(\hat{L}_i\) is naturally quantified by a divergence defined directly on this manifold. To this end, consider a twice continuously differentiable and strictly convex function \(h:\mathscr{P}(4)\rightarrow \mathbb{R}\). The associated Bregman divergence is defined as,
\begin{equation}\label{Bregman1}
d_h(L_i\|\hat{L}_i)
=
h(L_i)-h(\hat{L}_i)-\mathbb{G}_{(T_{\hat L_i}\mathscr{P}(4),\hat L_i)}
\left(
\nabla h(\hat L_i),
L_i-\hat L_i
\right),,
\end{equation}
with the following properties: 1) \(d_h(\hat{L}_i||L_i) \geq 0 \) with \(d_h(\hat{L}_i||L_i) = 0\) if and only if \(\hat{L}_i= L_i\), 2) \(d_h(\hat{L}_i||L_i) \) is strictly convex \cite{2019geometric}.
However, since it is not symmetric in general, it should be regarded as a pseudo-distance rather than a true metric. A particularly suitable choice on \(\mathscr{P}(4)\) is the log-determinant potential $h(L_i)=-\log|L_i|$. With this choice, the Bregman divergence becomes,
\begin{equation}\label{Bregman}
d_h(L_i\|\hat{L}_i)
=
\gamma\left(
\log\frac{|\hat{L}_i|}{|L_i|}
+\mathrm{tr}(\hat{L}_i^{-1}L_i)-4
\right),
\end{equation}
where \(\gamma>0\) is a scalar weighting factor. Equivalently, it may be written as
\begin{equation}
d_h(L_i\|\hat{L}_i)
=
\gamma\sum_{j=1}^{4}\left(-\log \lambda_j+\lambda_j-1\right),
\label{d_h lam}
\end{equation}
where \(\{\lambda_j\}_{j=1}^{4}\) are the eigenvalues of \(\hat{L}_i^{-1}L_i\). This divergence is affine-invariant and provides a second-order approximation of the affine-invariant Riemannian metric in \eqref{Metric}, making it a suitable candidate for the parameter error term in the Lyapunov analysis. The time derivative of \eqref{Bregman} can then be obtained as,
\begin{equation}\label{Ddot}
\dot{d}_h(L_i\|\hat{L}_i)
=
\gamma\,\mathrm{tr}\!\left(
\hat{L}_i^{-1}\dot{\hat{L}}_i\hat{L}_i^{-1}\tilde{L}_i
\right),
\end{equation}
with $\tilde{L}_i=\hat{L}_i-L_i$. Hence, the log-determinant Bregman divergence yields a valid Lyapunov-like parameter error measure on \(\mathscr{P}(4)\), while preserving the geometric structure of the pseudo-inertia manifold.
In the following, we will show that the update of pseudo-inertial matrix as,
\begin{equation}
    \Dot{\hat{L}}_i = \frac{1}{\gamma}\hat{L}_i\,\hat{\mathscr{L}}_i \hat{L}_i - \sigma(\hat{L}_i - L_i^0),
    \label{M_adapt}
\end{equation}
with \(L_i^0 \succ 0 \) being nominal value for the estimation, will keep the geometric property of the controller and adaptation consistent, yet requiring only one adaptation gain \(\gamma\) for the entire inter-connected rigid body system. Provided that the initial estimate satisfies \(\hat{L}_i(0)=L_i^0\succ 0\), the proposed update keeps the estimate on the SPD manifold and in a physically meaningful neighborhood of the nominal pseudo-inertia matrix. The leakage term further prevents parameter drift by pulling the estimate toward \(L_i^0\) in the absence of sufficiently informative excitation.

    \begin{lem}
        For the Bregman divergence defined in (\ref{Bregman}), the following upper bound for some constant \(C_i >0\) holds,
        \begin{equation}
        d_h(L_i\|\hat{L}_i)
        \;\le\;
        C_i\Big(
            \|E_i\|_{\hat{L}_i}^{2}
            + \|L_i - L_i^0\|_{\hat{L}_i}^{2}
        \Big).
        \label{d_h bound}
    \end{equation}
    \end{lem}
    \begin{pf}
        First, define the parameter estimation error and auxiliary variable as,
    \begin{equation}
        \Tilde{L}_i := \hat{L}_i - L_i \in\mathrm{Sym}(4), 
        \qquad
    E_i := \hat{L}_i - L_i^0 \in\mathrm{Sym}(4),
    \end{equation}
    from which,
    \begin{equation}
        \Tilde{L}_i = E_i - (L_i - L_i^0).
        \label{eq:Ltilde_decomp}
    \end{equation}
        Let us define the geodesic distance derived from metric (\ref{Metric}) as,
    \begin{equation}
    d_R(L,\hat{L})
    :=
    \left\|
    \log\!\big(L^{-\frac{1}{2}} \hat{L}\,L^{-\frac{1}{2}}\big)
    \right\|_{(T_I\mathscr{P}(4),I)},
    \label{eq:dR_def}
\end{equation}
    or
    \begin{equation}
        d_R(L,\hat{L})
        = \Bigg( \sum_{j=1}^4 \big(\log{\bar{\lambda}_j})^2 \Bigg)^{\frac{1}{2}},
        \label{d_R lam}
    \end{equation}
   with $\{\bar{\lambda}_j\}_{j=1}^4$ being the eigenvalues of \(L^{-1}\hat{L}\) with \({\lambda}_j = \bar{\lambda}_j^{-1}\). Now, for any compact set $\mathcal{K}\subset\mathscr{P}(4)$, by Taylor's approximation theorem applied to (\ref{d_h lam}) and (\ref{d_R lam}), there exists a constant $c>0$ such that
    \begin{equation}
        d_h(L\|\hat{L})
        \;\le\;
        c\, d_R(L,\hat{L})^{2},
        \qquad \forall\, L,\hat{L}\in \mathcal{K},
        \label{eq:local_equiv_dh_dR}
    \end{equation}
    where,
    \[
    c = \tfrac12 e^{\delta},
    \qquad
    \delta :=
    \sup_{L,\hat L\in\mathcal K}
    \lambda_{\max}\!\big(|\log(L^{-1}\hat L)|\big).
    \]
    Considering the square of triangle inequality for the distance metric (\ref{eq:dR_def}) according to Fig. \ref{fig:manifold_geometry} 
     and using $(a+b)^2 \le 2(a^2+b^2)$, the following upper bound can be established,
    \begin{equation}
        d_h(L_i\|\hat{L}_i)
        \;\le\;
        2c\Big(
            d_R(L_i,L_i^0)^{2}
            + d_R(L_i^0,\hat{L}_i)^{2}
        \Big).
        \label{eq:dh_bound_dR_terms}
    \end{equation}
Moreover, the affine-invariant distance (\ref{eq:dR_def}) is Lipschitz equivalent to the ambient norm on $\mathcal K$.
Hence, there exist constants $k_{C,i},k_{E,i}>0$ such that, for all
$L_i, L_i^0, \hat L_i\in\mathcal K$,
\begin{equation}
\begin{split}
    d_R(L_i,L_i^0)^{2}
    &\le
    k_{C,i}\,\|L_i - L_i^0\|_{(T_{\hat{L}_i}\mathscr{P}(4),\hat{L}_i)}^{2},\\
    d_R(L_i^0,\hat{L}_i)^{2}
    &\le
    k_{E,i}\,\|\hat{L}_i - L_i^0\|_{(T_{\hat{L}_i}\mathscr{P}(4),\hat{L}_i)}^{2}.
\end{split}
\label{upper_dR}
\end{equation}
Finally, substituting (\ref{upper_dR}) into \eqref{eq:dh_bound_dR_terms} will yield to (\ref{d_h bound}) with \(C_i := 2c\,\max\{k_{C,i},k_{E,i}\}\).
    \end{pf}
    
    To facilitate the subsequent stability analysis, we define the following equivalent form,
    \begin{equation}
    \begin{split}
        &\hat{M}_i\,\mathscr{A}^r_{i} -ad^\top_{\mathscr{V}_i}\,(\hat{M}_i\,\mathscr{V}^r_{i}) =  tr(\hat{L}_i \hat{\mathscr{L}}_i),\\
        &M_i\,\mathscr{A}_i -ad^\top_{\mathscr{V}_i}\,(M_i\,\mathscr{V}_i)= tr({L}_i \mathscr{L}_i),
    \end{split}
        \label{P3}
    \end{equation}
    where \(\mathscr{L}_i\) is a unique symmetric matrix, defined in \cite{2019geometric}. On the other hand, by subtracting (\ref{Fi}) from (\ref{Fr_ad}) along with adding \(tr(L_i \mathscr{L}_i)-tr(L_i \mathscr{L}_i)\) and using (\ref{P3}), we have,
    \begin{equation}
        \begin{split}
        {M}_i (\mathscr{A}^r_{i}-\mathscr{A}_{i}) &= ({F}^r_{{i}} - {F}_{{i}}) - tr(\Tilde{L}_i \hat{\mathscr{L}}_i) +\,Ad_{\mathcal{T}^{-1}_{i,i+1}}^\top\, (F^r_{{i+1}}-F_{{i+1}})+ad^\top_{\mathscr{V}_i}\,({M}_i\,(\mathscr{V}^r_{i}-\mathscr{V}_i))-K_v(\mathscr{V}^r_{i}-\mathscr{V}_{i}).\\
    \end{split}
    \label{F_til}
    \end{equation}
\begin{thm}
Consider a complex robotic system decomposed into the subsystems, as shown in Fig. \ref{fig:interconnected}, with the rigid body an joint dynamics of (\ref{Fi}) and (\ref{Ji}) in the presence of parameter uncertainties. Each subsystem is virtually stable in the sense of Definition \ref{VS} under the adaptive modular geometric control (\ref{Fr_ad}) and (\ref{Jri}) and adaptation law (\ref{M_adapt}). 
\end{thm}
\begin{pf}
    Define the local Lyapunov function as,
    \begin{equation}
        \mathrm{v}_i = \mathrm{V}_i+d_h(L_i||\hat{L}_i).
        \label{vi}
    \end{equation}
    First we investigate the time derivative of \(d_h(L_i||\hat{L}_i)\) in (\ref{Ddot}). 
\begin{figure}[pos=t]
    \centering
    \begin{tikzpicture}[scale=1.25]

    \shade[draw=blue!35, top color=blue!12, bottom color=blue!3]
        (-3,0) .. controls (-2,1.6) and (2,1.6) .. (3,0)
        .. controls (2,-1.3) and (-2,-1.3) .. cycle;

    \draw[blue!25] (-2.2,0.05) .. controls (-1.4,0.9) and (1.4,0.9) .. (2.2,0.05);
    \draw[blue!18] (-1.3,0.05) .. controls (-0.7,0.45) and (0.7,0.45) .. (1.3,0.05);

    \node at (2.55,1.05) {$\mathscr{P}(4)$};

    \coordinate (L)    at (-1.6,-0.15);
    \coordinate (L0)   at (-0.1,0.82);
    \coordinate (Lhat) at (1.45,-0.05);

    \draw[thick, fill=black!8] (L) ellipse [x radius=0.34, y radius=0.18, rotate=25];
    \draw[thick, dashed, fill=black!4] (L0) ellipse [x radius=0.28, y radius=0.14, rotate=-15];
    \draw[thick, fill=black!12] (Lhat) ellipse [x radius=0.38, y radius=0.16, rotate=8];

    \node[below left] at (-1.1,-0.15) {$L_i$};
    \node[above] at (0.1,0.82) {$L_i^0$};
    \node[below right] at (1.45,-0.15) {$\hat{L}_i$};

    \draw[thick] (L) to[out=10,in=170] (Lhat);
    \draw[dashed, thick] (L) to[out=55,in=210] (L0);
    \draw[dashed, thick] (L0) to[out=330,in=120] (Lhat);

    \node at (-0.05,-0.15) {$d_R(L_i,\hat{L}_i)$};
    \node at (-1.1,0.8) {$d_R(L_i,L_i^0)$};
    \node at (1.2,0.7) {$d_R(L_i^0,\hat{L}_i)$};

    \end{tikzpicture}
    \caption{Schematic illustration of the SPD manifold $\mathscr{P}(4)$. The true inertia $L_i$, nominal model $L_i^0$, and estimate $\hat{L}_i$ are depicted by local ellipsoidal proxies associated with symmetric positive-definite matrices, while the curves represent Riemannian geodesics between them.}
    \label{fig:manifold_geometry}
\end{figure}
By substituting (\ref{M_adapt}) in (\ref{Ddot}) and using (\ref{eq:Ltilde_decomp}), we have,
    \begin{align}
    \Dot{d}_h
    &= \mathrm{tr}(\hat{\mathscr{L}}_i\Tilde{L}_i)
    - \gamma\sigma\,
        \mathrm{tr}\!\left(
            \hat{L}_i^{-1}E_i\hat{L}_i^{-1}\Tilde{L}_i
        \right),
    \label{eq:dh_raw}
    \end{align}
    with the second term,
    \begin{equation}
        \mathrm{tr}(\hat{L}_i^{-1}E_i\hat{L}_i^{-1}\Tilde{L}_i)
        = 2\, \mathbb{G}_{(T_{\hat{L}_i}\mathscr{P}(4),\hat{L}_i)}(E_i,\Tilde{L}_i).
        \label{E_tr}
    \end{equation}
    Substituting \eqref{eq:Ltilde_decomp} into (\ref{E_tr}) in the sense of (\ref{Metric}), we have,
    \begin{equation}
        \begin{split}
            \mathbb{G}_{(T_{\hat{L}_i}\mathscr{P}(4),\hat{L}_i)}(E_i,\Tilde{L}_i)=
        \|E_i\|_{(T_{\hat{L}_i}\mathscr{P}(4),\hat{L}_i)}^2-
       \mathbb{G}_{(T_{\hat{L}_i}\mathscr{P}(4),\hat{L}_i)}(E_i,L_i - L_i^0).
        \end{split}
        \label{E3}
    \end{equation}
    Now, by applying Young’s inequality to the second term of (\ref{E3}) for any $\beta_i>0$ and substituting all in \eqref{eq:dh_raw} yields,
    \begin{equation}
        \begin{split}
            \Dot{d}_h
        &\le 
        \mathrm{tr}(\hat{\mathscr{L}}_i\Tilde{L}_i)-\gamma\sigma(1-\beta_i)\|E_i\|_{(T_{\hat{L}_i}\mathscr{P}(4),\hat{L}_i)}^2+ \frac{\gamma\sigma}{4\beta_i}\|L_i-L_i^0\|_{(T_{\hat{L}_i}\mathscr{P}(4),\hat{L}_i)}^{2}.
        \end{split}
        \label{d_h_dot}
    \end{equation}
    
    Now, by replacing (\ref{d_h_dot}) in the time derivative of (\ref{vi}) and using (\ref{F_til}), (\ref{V}), and (\ref{d_h bound}) along with following the same procedure in previous section such as (\ref{dV1}), (\ref{dtPsi}), and  (\ref{ad_bound}), we have,
    \begin{equation}
        \begin{split}
            \Dot{\mathrm{v}}_i & \leq - \Omega_i{\mathrm{v}}_i + \varepsilon_i,
            \label{v_i end}
        \end{split}
    \end{equation}
    with \(\Omega_i = \min({\varrho}_i,\varpi_iC_i^{-1})/\bar{\rho}_i\), \(\bar{\rho} = \max{(\rho_i,1)}\), \(\varepsilon_i = 2_i \varpi_i \|L_i-L_i^0\|_{\hat{L}_i}^{2}\), \(\varpi_i = \gamma\sigma_i \min((1-\beta_i),1/(4\beta_i))\). Now, by considering the procedure for the joint subsystems as in (\ref{Va})-(\ref{DVa}), defining the total Lyapunov function as \(\mathrm{v}_T = \sum_1^n (\mathrm{V}_{ai} + \mathrm{v}_i)\), taking its time derivative and replacing from (\ref{v_i end}), one can obtain,
    \begin{equation}
        \begin{split}
            \Dot{\mathrm{v}}_T \leq - \upsilon \mathrm{v}_T+\varepsilon,
        \end{split}
    \end{equation}
    with \(\upsilon = \min(\Omega_i,{2k_{ai}}/{I_{mi}})\) and \(\varepsilon = \min(\varepsilon_i)\). Then, the following can be obtained,
\begin{equation}
\label{eq:vT_bound}
\mathrm v_T(t)
\;\le\;
\mathrm v_T(0)e^{-\upsilon t}
\;+\;
\frac{\varepsilon}{\upsilon}\left(1-e^{-\upsilon t}\right) \leq \mathrm v_T(0) + \frac{\varepsilon}{\upsilon}.
\end{equation}
Therefore, all closed-loop trajectories remain bounded for all $t\ge0$, and the closed-loop tracking and estimation errors are semi-globally uniformly ultimately bounded with respect to an ultimate bound proportional to $\varepsilon/\upsilon$. Additionally, by defining the following according to (\ref{d_h lam}),
\begin{equation}
    \phi(\lambda) := -\log \lambda + \lambda - 1, \qquad \lambda>0,
    \label{phi}
\end{equation}
the following can be shown as a result of (\ref{eq:vT_bound}) and (\ref{vi}),
\begin{equation}
\label{eq:lambda_max_final}
\lambda_{\max}(\hat L_i^{-1} L_i) \;\le\; \overline{\lambda} = \phi^{-1}(\mathrm v_T(0) + \frac{\varepsilon}{\upsilon}),
\end{equation}
demonstrating the boundedness of the parameter estimation error on the manifold. Furthermore, following the same procedure in the last section, the upper bound for the \(\sup_{t>0}||\mathscr{V}_{i}||\) can be obtained. It must be noted that by proper initialization that satisfies \(\mathrm v_T(0)\leq -\frac{\varepsilon}{\upsilon} + \lambda_{min}(K_{z_i}) \pi^2/2\), the local configurational error \(\eta_i\) will remain in non-singular region with \(trace(R_{e_i}) \geq -1\).
\end{pf}

\begin{figure}[pos=t]
    \centering
    \includegraphics[width=0.5\linewidth]{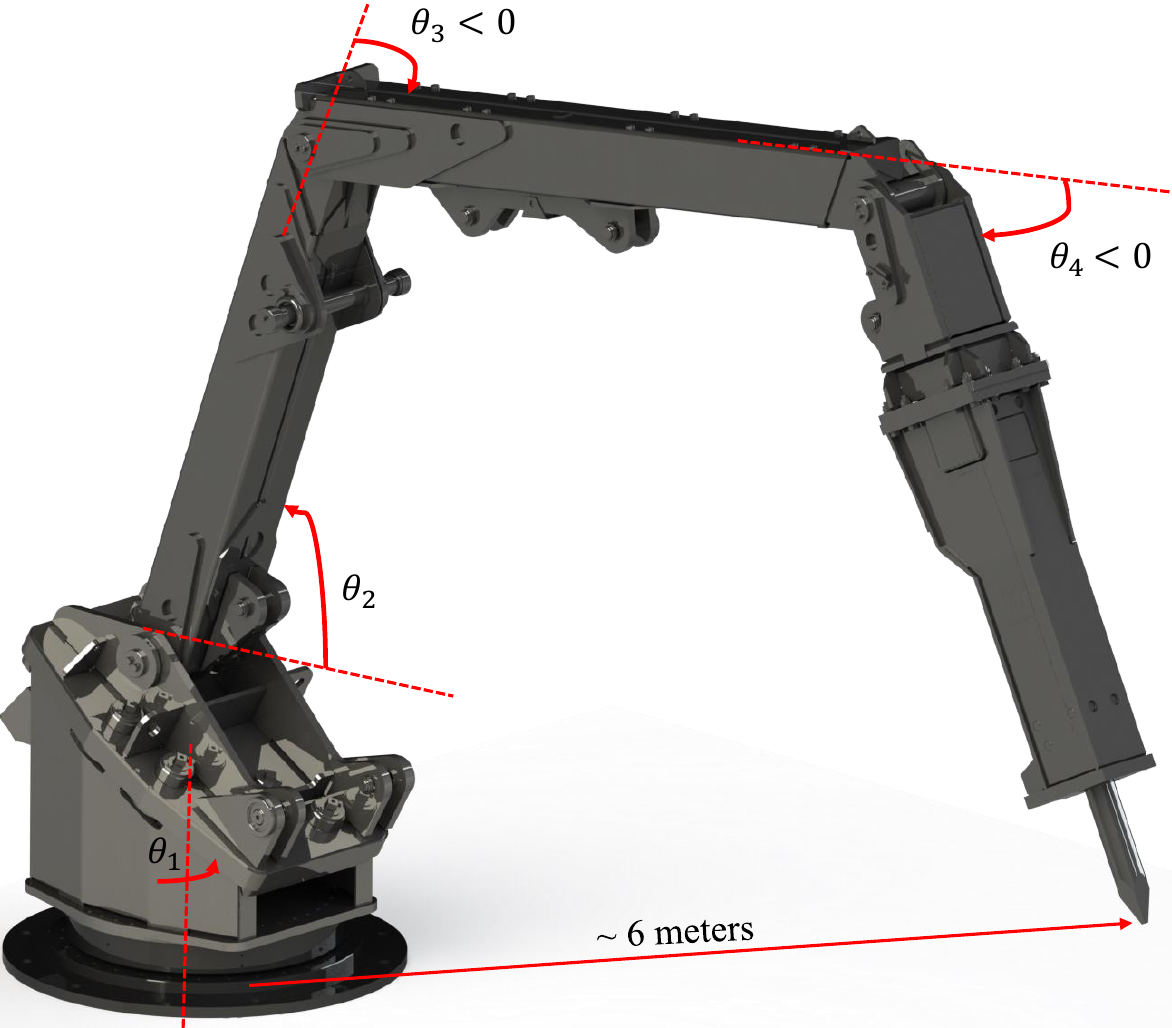}
    \caption{4R redundant heavy-duty manipulator scheme employed in simulations. The total mass is $\approx 9350~\mathrm{kg}$, reach $\approx 7~\mathrm{m}$}
    \label{Render}
\end{figure}

\section{Results}
\label{Results}
In this section, numerical simulations are conducted on a 4R redundant heavy-duty manipulator (total mass $\approx 9350~\mathrm{kg}$, reach $\approx 7~\mathrm{m}$; see Fig.~\ref{Render}) in order to evaluate the proposed control framework. The kinematic redundancy is exploited to regulate the end-effector orientation. The manipulator model is implemented in MATLAB/Simulink with a sampling frequency of $1~\mathrm{kHz}$. The proposed modular geometric controller (MGC) (\ref{Fri}) is compared against the modular controller of~\cite{humaloja2021decentralized} (VDC) and the geometric impedance controller of~\cite[equation~(32)]{seo2023geometric} (GIC). Furthermore, the adaptive extension of the proposed method (AMGC) (\ref{Fr_ad}) is compared with the MGC under parametric uncertainty to showcase the performance of the proposed adaptive control. All simulations are performed using the actual system parameters so as to preserve fidelity with respect to the real robotic platform. The use of an extreme-scale manipulator, characterized by substantial inertia and a long kinematic reach, is intentional: such a system provides a demanding test environment in which dynamic coupling effects are pronounced and even modest relative tracking errors may correspond to significant absolute deviations at the end-effector.

The gains of the proposed controllers are selected as \(
\Gamma_i = \mathrm{diag}(30,30,30,20,0.1,10), \,
K_v = 2000\,\cdot I_6, \,
\gamma = 10^4, \, \sigma = 0.1.\)
The gains of all benchmark controllers were manually tuned to provide competitive tracking performance under comparable control effort. For the comparison of the non-adaptive controllers, the initial joint configuration is chosen as \(
\boldsymbol{\theta}(0)=[-10,\;70,\;-90,\;-70]^\top \mathrm{deg}, 
\,
\dot{\boldsymbol{\theta}}(0)=\boldsymbol{0},\)
so as to induce a sufficiently large initial configuration error. This choice is intended to make the transient response clearly visible and to better reveal the effect of the non-Euclidean geometric formulation during large-motion evolution on the manifold. For the adaptive-control study, the initial configuration is chosen to coincide with the desired one, while a $\pm 10\%$ uncertainty is introduced into the rigid-body densities in Simscape. This parametric perturbation also provides a meaningful test of the robustness of the proposed controller.

In all the simulations, the desired end-effector position trajectory is prescribed as,
\[
\mathcal{P}^d =
\begin{bmatrix}
x^d\\
y^d\\
z^d
\end{bmatrix}
=
R_x(\pi/6)
\begin{bmatrix}
5 + 0.5\cos(2\pi t/T)\\
0\\
-0.5 + 0.5\cos(2\pi t/T)
\end{bmatrix},
\]
which corresponds to a tilted circular trajectory in Cartesian space with period $T=60~\mathrm{s}$. In addition, the end-effector orientation, denoted by $\zeta$, is prescribed such that the end-effector remains perpendicular to the $xy$-plane, namely $\zeta_d=-\pi/2$. Further, the end-effector position and orientation error shown in the figures are defined as \(\mathbf{e}^p(t) = \mathcal{P}^d(t)-\mathcal{P}(t)\) and \(e^o(t) = \zeta^d(t)-\zeta(t)\), respectively. Without loss of generality, the effect of the joint dynamics is neglected in this simulation study. Since the GIC benchmark does not explicitly account for joint dynamics, the joint-level gain is set to $k_a=0$ in order to ensure a fair comparison at the level of rigid-body dynamics alone. This assumption is justified by the comparatively small influence of the joint dynamics relative to the dominant rigid-body dynamics of the considered heavy-duty system. Moreover, if the proposed approach demonstrates superior performance even in the absence of the joint module, then incorporating the joint-level control action is expected only to further improve the closed-loop response.

\subsection{Performance Comparison}
Fig.~\ref{Path_track}(a)--(c) illustrate the position-tracking performance of the considered controllers, whereas Fig.~\ref{Path_track}(d) reports the corresponding end-effector orientation regulation. It is observed that both geometric controllers, namely MGC and GIC, outperform the Euclidean-space-based modular controller during the transient phase and exhibit smoother convergence toward the desired path. The proposed MGC, however, demonstrates superior overall performance, particularly in the orientation response shown in Fig.~\ref{Path_track}(d), where it achieves smooth convergence without noticeable oscillation in the vicinity of the desired equilibrium. To provide a more intuitive illustration of the tracking behavior, Fig.~\ref{circle} shows the end-effector trajectories generated by the different controllers, while Fig.~\ref{fig:controller_snaps} presents snapshots of the corresponding robot motion. As seen from these results, both geometric methods converge smoothly to the desired path, with the proposed MGC exhibiting the fastest transient response with accurate orientation regulation.

\begin{figure}[pos=t]
    \centering
    \includegraphics[width=0.7\linewidth]{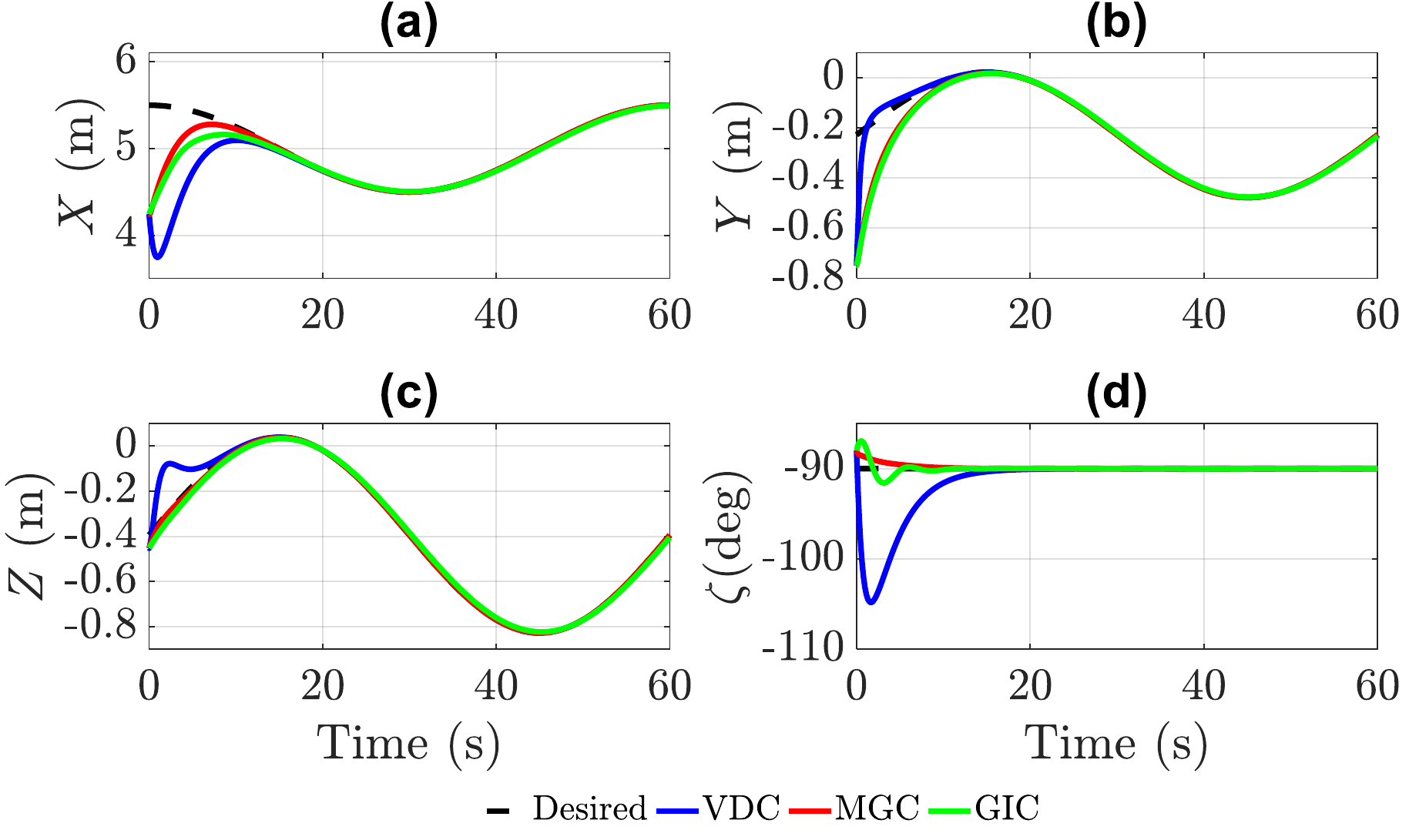}
    \caption{End-effector tracking performance under the considered controllers. (a)--(c) report the Cartesian position responses in the $x$-, $y$-, and $z$-directions, respectively, while (d) shows the corresponding end-effector orientation regulation.}
    \label{Path_track}
\end{figure}
\begin{figure}[pos=t]
    \centering
    \includegraphics[width=0.7\linewidth]{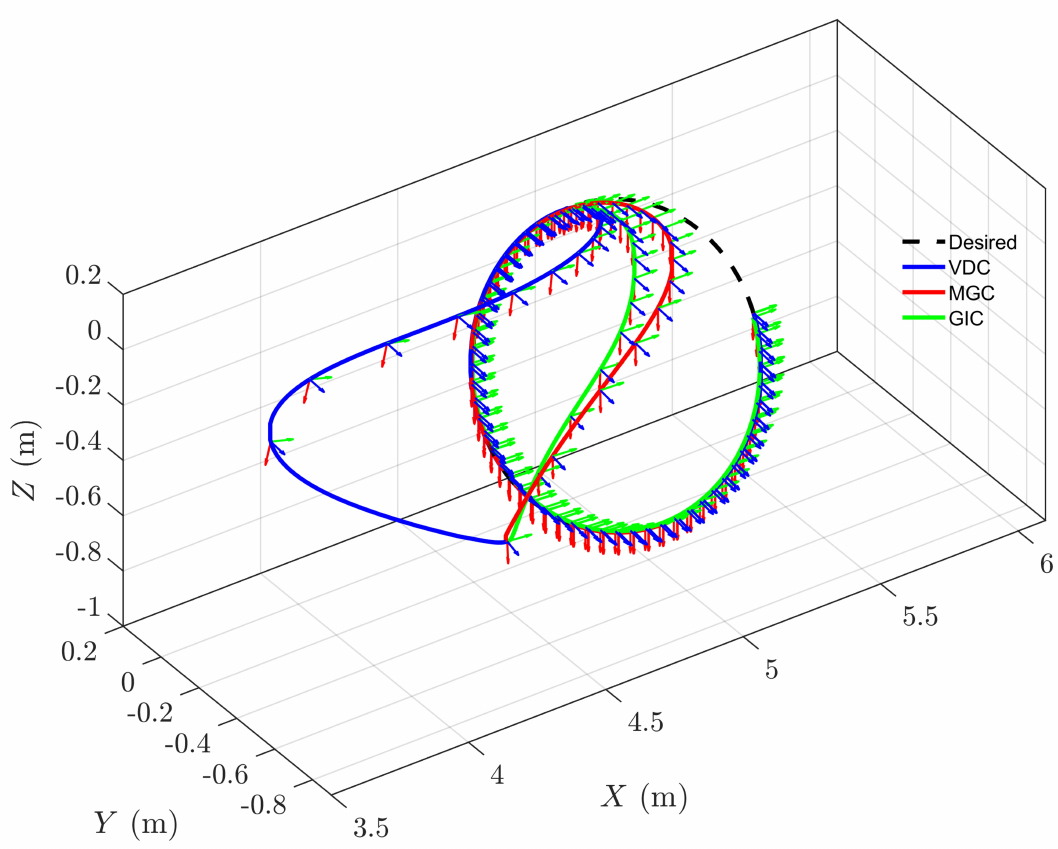}
    \caption{Three-dimensional end-effector trajectories generated by the considered controllers, together with the attached end-effector coordinate frames. The figure illustrates the transient convergence and steady-state tracking behavior along the prescribed spatial path.}
    \label{circle}
\end{figure}

\begin{figure}[pos=t]
    \centering
    \subfloat[VDC\label{fig:VDC}]{
        \includegraphics[width=0.32\linewidth]{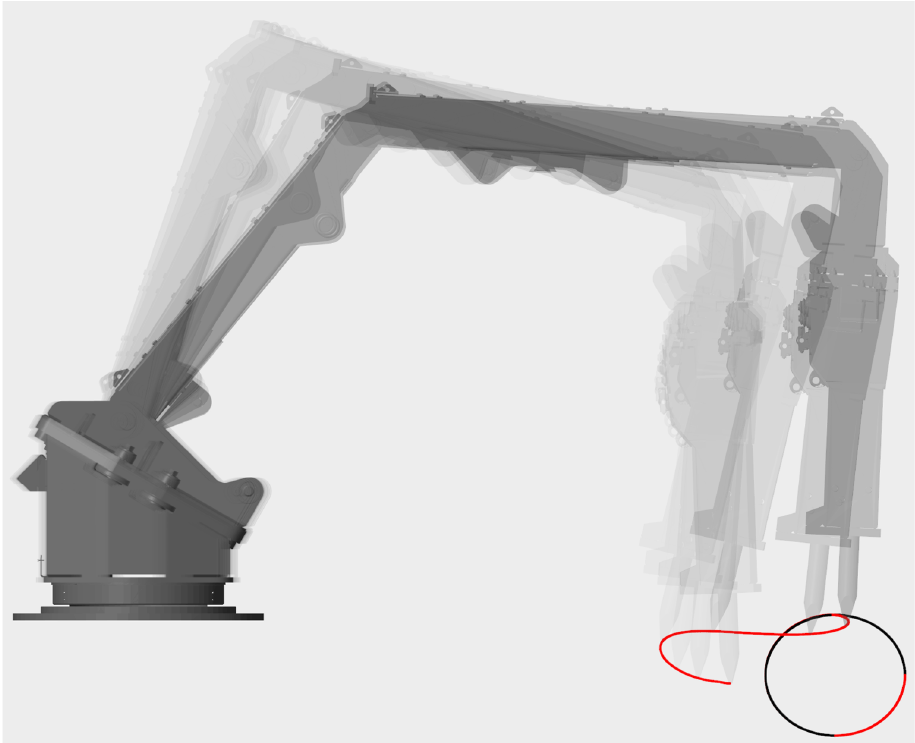}
    }\hfill
    \subfloat[GIC\label{fig:GIC}]{
        \includegraphics[width=0.32\linewidth]{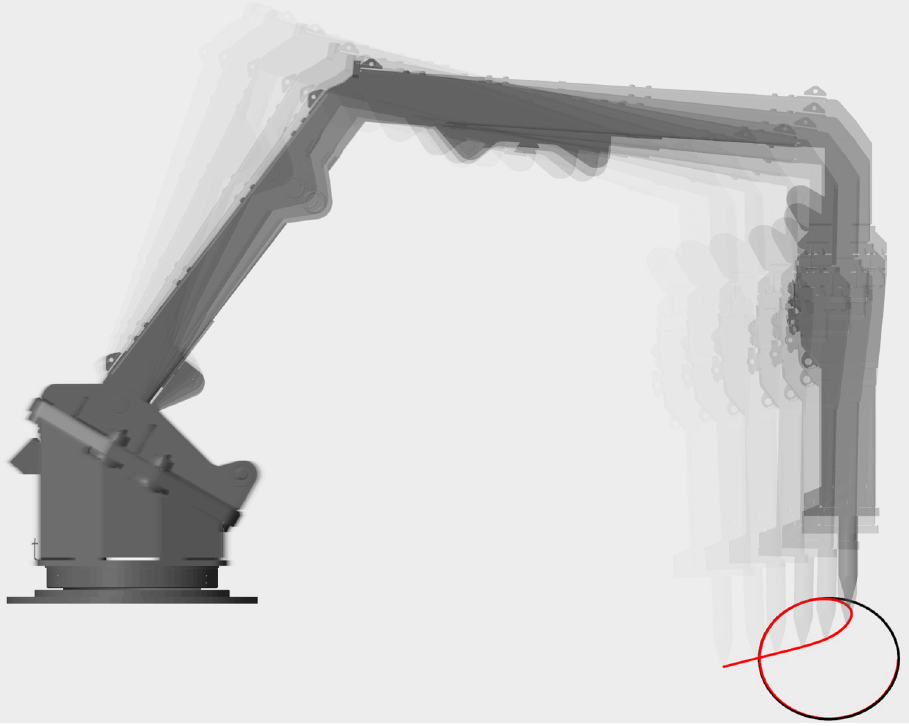}
    }\hfill
    \subfloat[MGC\label{fig:MGC}]{
        \includegraphics[width=0.32\linewidth]{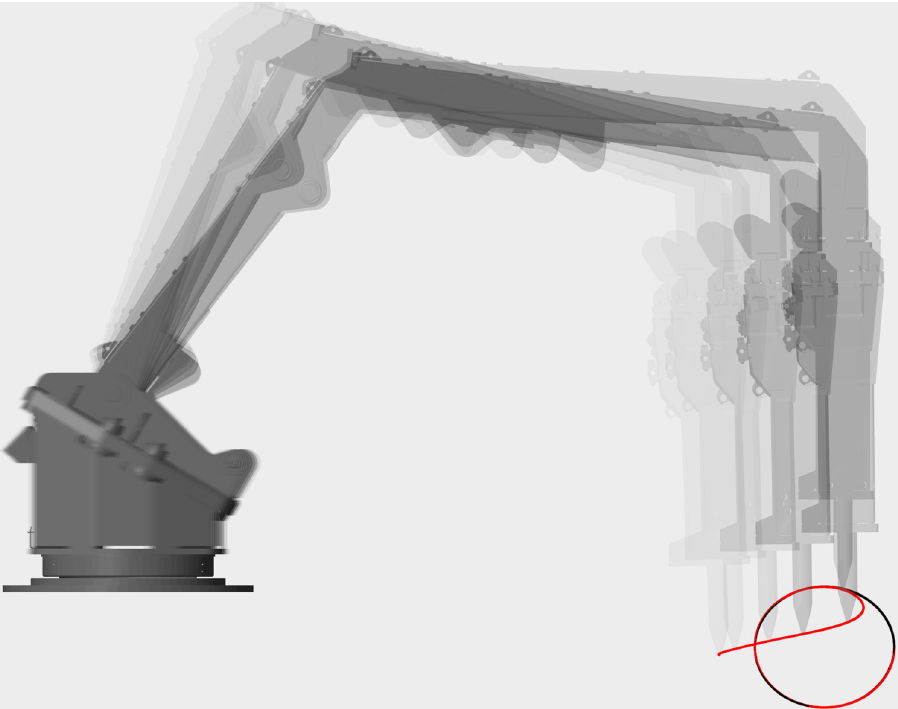}
    }
    \caption{Snapshots of the manipulator motion under the considered non-adaptive controllers.}
    \label{fig:controller_snaps}
\end{figure}

\begin{figure}[pos=t]
    \centering
    \includegraphics[width=0.7\linewidth]{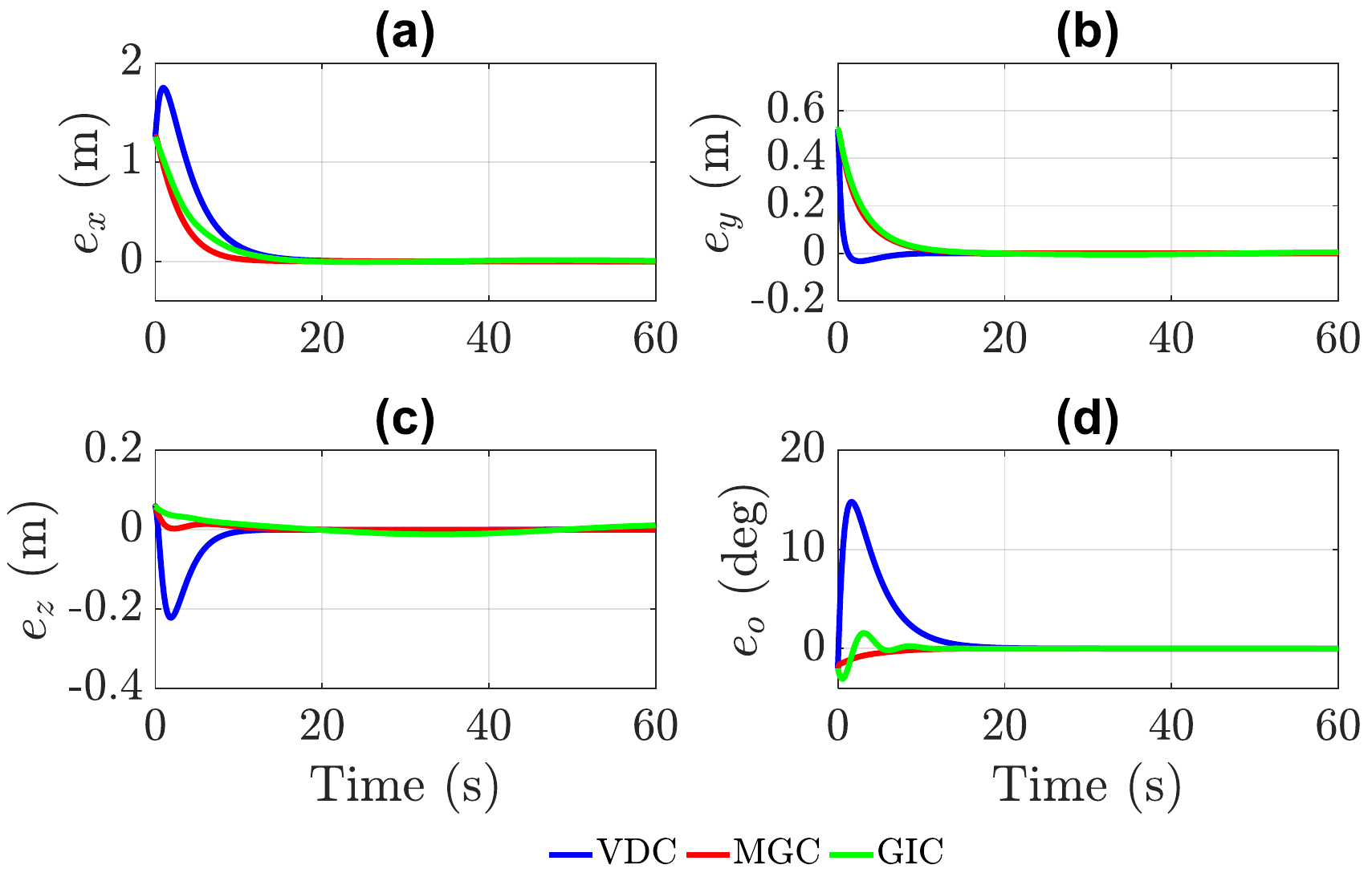}
    \caption{Time histories of the end-effector tracking errors. (a)--(c) correspond to the Cartesian position errors in the $x$-, $y$-, and $z$-directions, respectively, and (d) shows the orientation error.}
    \label{tracking_error}
\end{figure}

\begin{figure}[pos=t]
    \centering
    \includegraphics[width=0.8\linewidth]{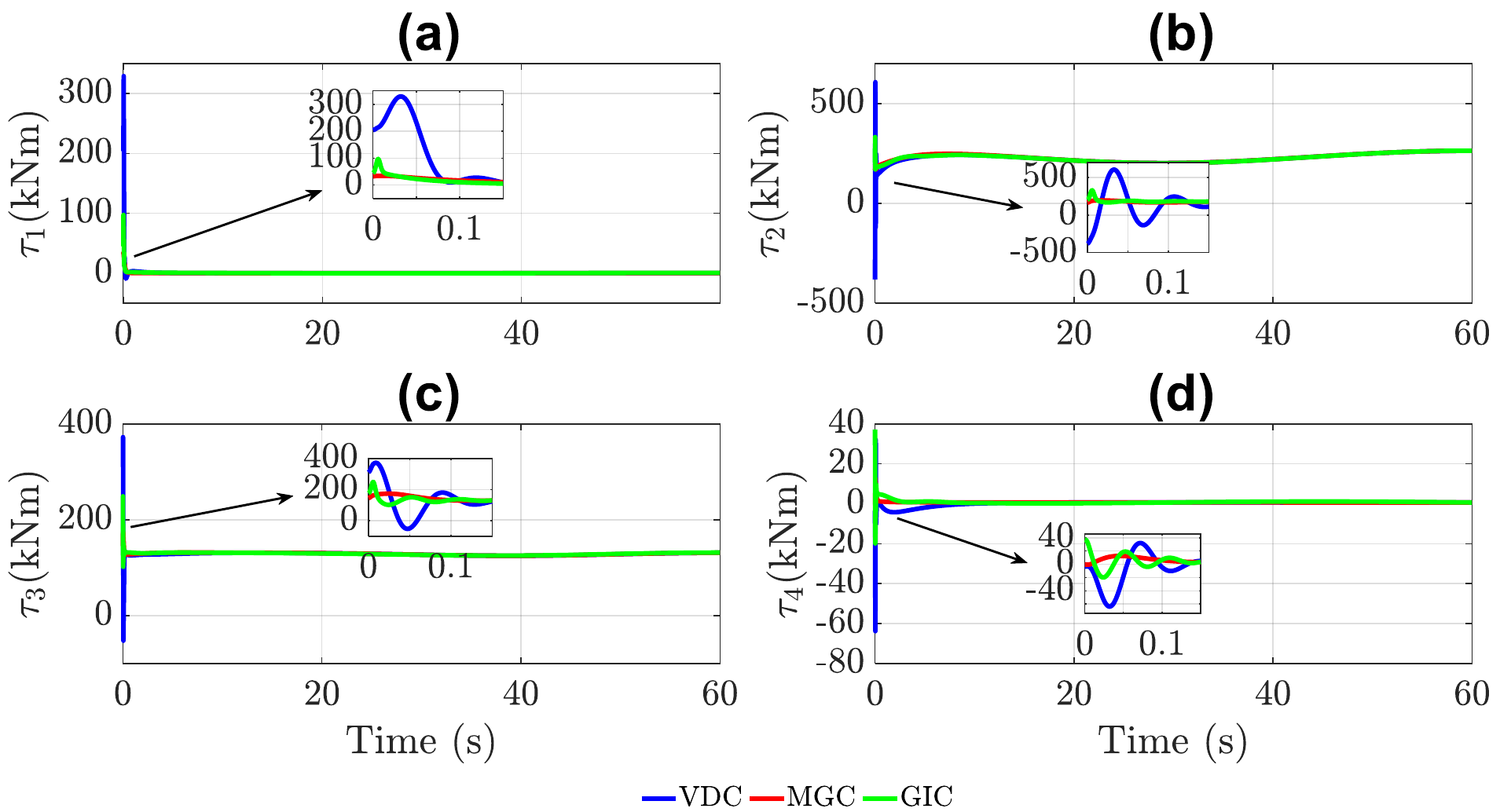}
    \caption{Joint control torques produced by the considered controllers. (a)--(d) correspond to joints 1--4, respectively. The results indicate that the proposed MGC yields smoother torque profiles while requiring control effort comparable to, or lower than, the benchmark methods.}
    \label{torques}
\end{figure}

Fig.~\ref{tracking_error} further presents the time histories of the tracking errors. The results indicate that the proposed controller not only improves substantially upon the Euclidean-space-based modular controller, but also achieves improved orientation regulation relative to its geometric benchmark counterpart. Notably, this performance improvement is obtained without requiring additional control effort. Indeed, Fig.~\ref{torques} shows that the proposed controller demands less control input than VDC and a control effort comparable to that of GIC. Collectively, these results demonstrate that the proposed method achieves faster and smoother convergence than the considered state-of-the-art benchmark controllers while maintaining a comparable, or lower, level of control effort.

\subsection{Adaptive Control Performance}
The tracking performance of the AMGC and MGC controllers in the presence of parametric uncertainty is shown in Figs.~\ref{er_10m}--\ref{fig:L_hat}. Figs~\ref{er_10m}(a)--(c) present the position tracking errors, while Fig.~\ref{er_10m}(d) shows the orientation regulation performance under $-10\%$ uncertainty. Likewise, Fig.~\ref{er_10p} illustrates the corresponding tracking performance under $+10\%$ uncertainty. The results show that AMGC provides a clear performance improvement over MGC in the presence of parametric uncertainty. Given the scale of the manipulator, a $\pm 10\%$ variation in the density of the rigid bodies constitutes a severe perturbation and thus defines a demanding evaluation scenario. Nevertheless, while AMGC achieves superior tracking performance by effectively compensating for these uncertainties, MGC still preserves bounded closed-loop behavior and maintains system stability. Furthermore, Figs.~\ref{L_hat_10m} and~\ref{L_hat_10p} show the time histories of the norms of the estimated parameters over a longer simulation horizon ($T=120~\mathrm{s}$), in order to better illustrate the boundedness and convergence properties of the adaptive signals. Taken together, these results validate the theoretical claims regarding stability and tracking performance.

\begin{figure}[pos=t]
    \centering
    \includegraphics[width=0.7\linewidth]{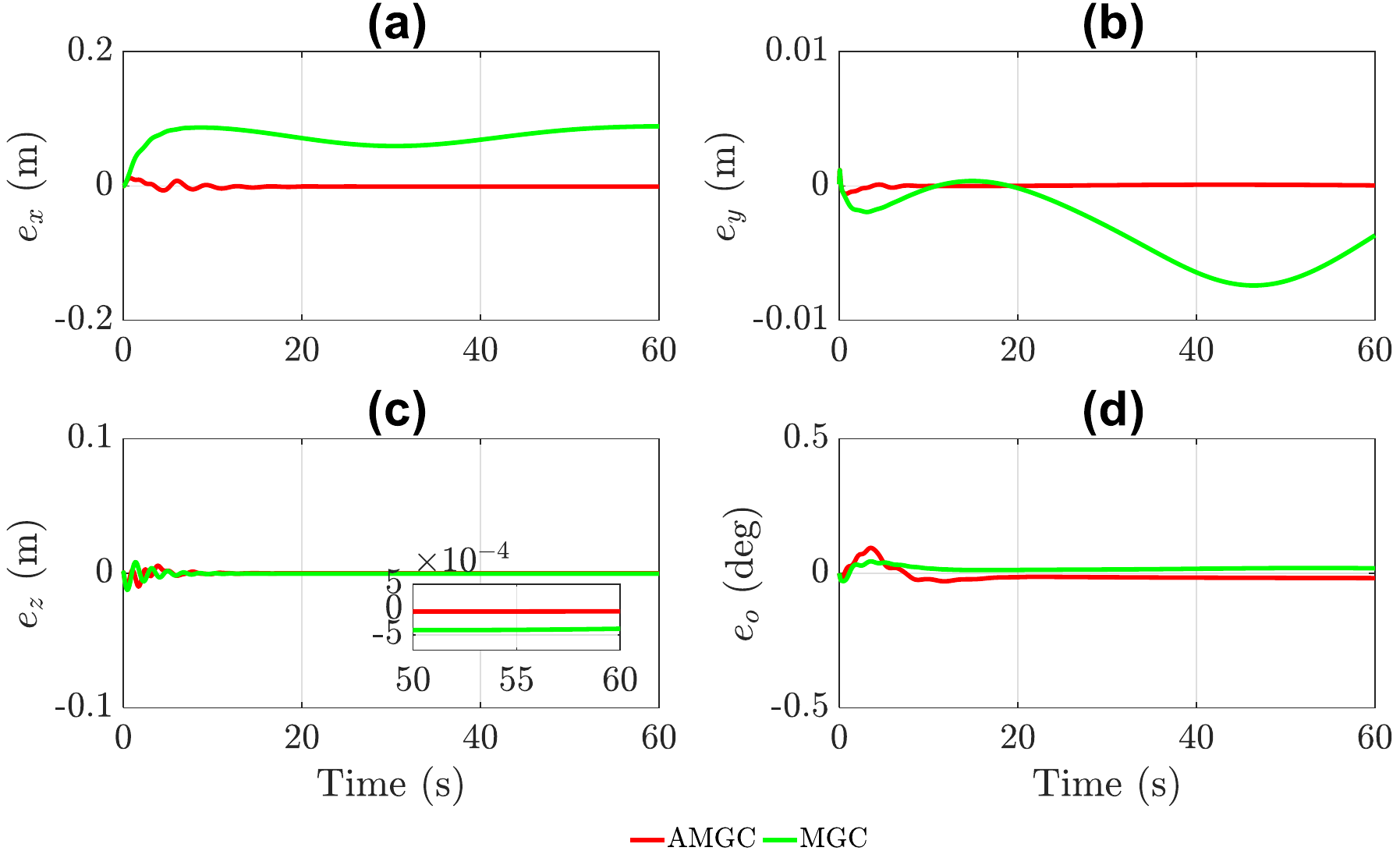}
    \caption{Comparison of the tracking performance of the MGC and AMGC schemes under $-10\%$ parametric uncertainty.}
    \label{er_10m}
\end{figure}

\begin{figure}[pos=t]
    \centering
    \includegraphics[width=0.7\linewidth]{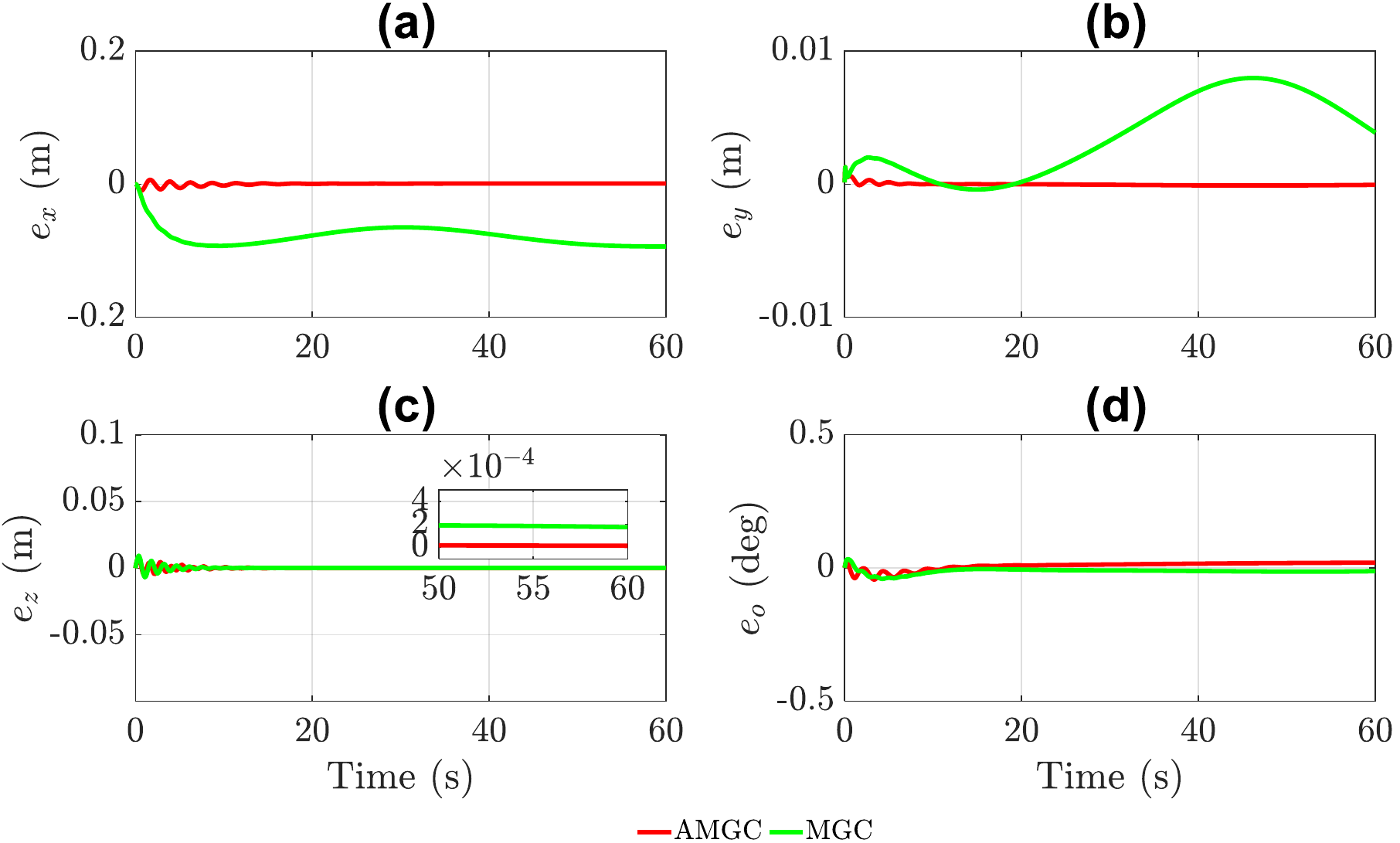}
    \caption{Comparison of the tracking performance of the MGC and AMGC schemes under $+10\%$ parametric uncertainty.}
    \label{er_10p}
\end{figure}

\begin{figure}[pos=t]
    \centering
    \subfloat[\label{L_hat_10m}]{
        \includegraphics[width=0.45\linewidth]{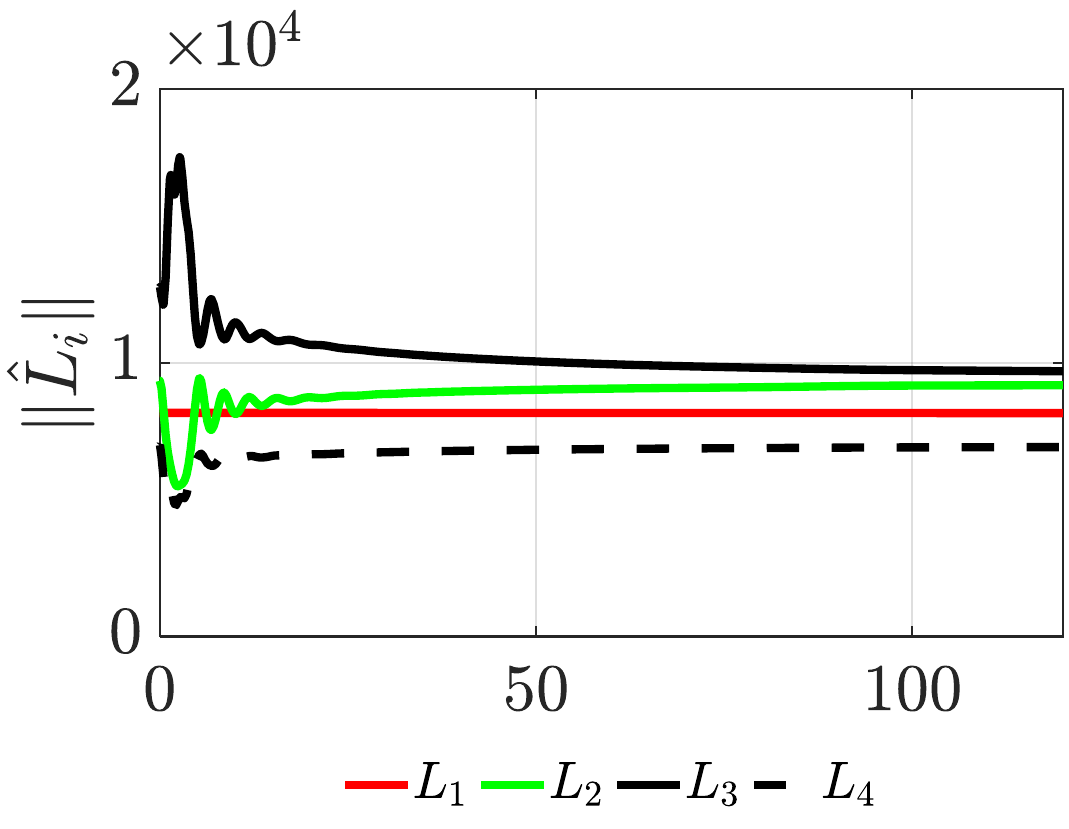}
    }
    \hfill
    \subfloat[\label{L_hat_10p}]{
        \includegraphics[width=0.45\linewidth]{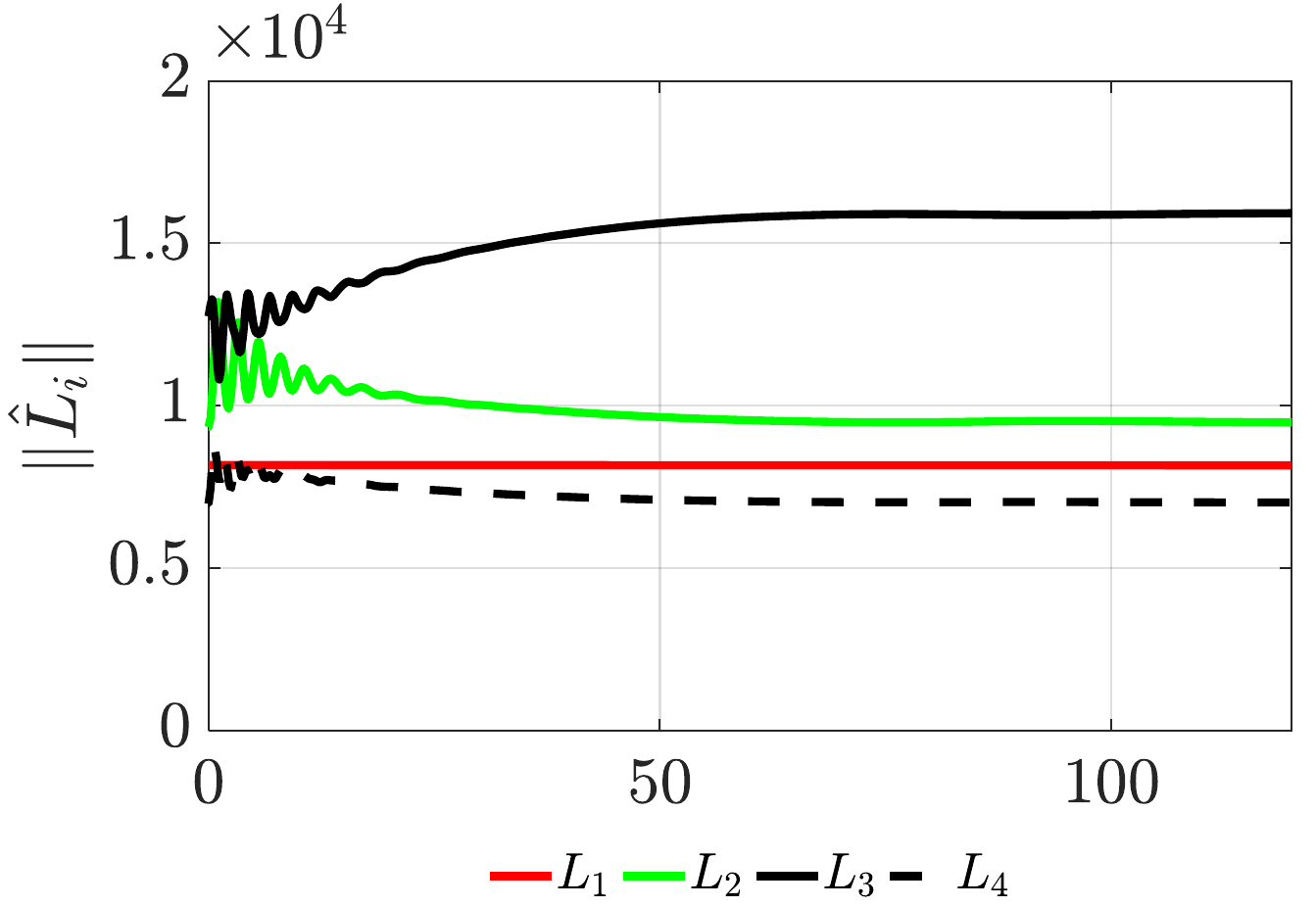}
    }
    \caption{Time histories of the estimated parameters under a) $+10\%$ parametric uncertainty and b) $-10\%$ parametric uncertainty. Each curve represents the norm of the corresponding parameter estimate.}
    \label{fig:L_hat}
\end{figure}

\subsection{Quantitative Analysis}
Finally, to provide a more comprehensive assessment of controller performance, both tracking accuracy and control effort are quantified. The adopted performance measures for position tracking include the RMS and maximum norms of the position error, $\|\mathbf{e}^p(t)\|$, denoted by $e_{\mathrm{rms}}^p$ and $e_{\max}^p$, respectively, together with the steady-state error $e_{\mathrm{ss}}^p$ and the integrated absolute error $e_{\mathrm{IAE}}^p$. The orientation performance is evaluated by the RMS and maximum orientation errors, denoted by $e_{\mathrm{rms}}^o$ and $e_{\max}^o$, respectively. In addition, the control effort is characterized by the overall RMS torque, $\tau_{\mathrm{rms}}$, computed over all joints.

\begin{table*}[t]
\centering
\caption{Performance comparison metrics}
\label{tab:controller_metrics}
\begin{threeparttable}
\setlength{\tabcolsep}{6pt}
\renewcommand{\arraystretch}{1.18}
\begin{adjustbox}{width=\textwidth}
\begin{tabular}{llcccccccc}
\toprule
\rowcolor{gray!15}
\textbf{Case} & \textbf{Controller} & $\mathbf{e}_{\mathrm{rms}}^p$ & $\mathbf{e}_{\max}^p$ & $\mathbf{e}_{\mathrm{ss}}^p$ & $\mathbf{e}_{\mathrm{IAE}}^p$ & $e_{\mathrm{rms}}^o$ & $e_{\max}^o$ & $\tau_{\mathrm{rms}}$ \\
\midrule

\multirow{3}{*}{\makecell[l]{Non-adaptive}}
& MGC & 0.2199 & 1.3628 & 0.0001 & 3.9828 & 0.2933 & 2.0000 & 263.2088 \\
& VDC \cite{humaloja2021decentralized} & 0.4131 & 1.7584 & 0.0015 & 9.0244 & 3.5395 & 14.8169 & 261.8431 \\
& GIC \cite{seo2023geometric} & 0.2505 & 1.3599 & 0.0132 & 5.7933 & 0.4791 & 3.0584 & 262.2216\\
\midrule

\multirow{2}{*}{\makecell[l]{Adaptive case\\ $-10\%$ uncertainty}}
& AMGC & 0.0029 & 0.0138 & 0.0011 & 0.1094 & 0.0247 & 0.0942 & 240.1292\\
& MGC  & 0.0747 & 0.0889 & 0.0889 & 4.4125 & 0.0191 & 0.0440 & 236.4682\\
\midrule

\multirow{2}{*}{\makecell[l]{Adaptive case\\ $+10\%$ uncertainty}}
& AMGC & 0.0021 & 0.0109 & 0.0012 & 0.0858 & 0.0158 & 0.0443 & 293.4815\\
& MGC  & 0.0801 & 0.0937 & 0.0937 & 4.7299 & 0.0156 & 0.0415 & 298.3060\\
\bottomrule
\end{tabular}
\end{adjustbox}
\begin{tablenotes}[flushleft]
\footnotesize
\item The units of the reported metrics are degrees for orientation, meters for position, and kN$\cdot$m for torque.
\end{tablenotes}
\end{threeparttable}
\end{table*}

The quantitative results of all simulation studies are summarized in Table~\ref{tab:controller_metrics}. For the non-adaptive case, the geometric controllers exhibit clear performance advantages over VDC across the reported metrics. This is particularly evident in $e_{\max}^p$ and $e_{\max}^o$, where VDC shows larger transient errors due to its pronounced initial overshoot. Among all non-adaptive methods, the proposed MGC achieves the best overall performance in both position and orientation regulation. More specifically, MGC reduces the RMS position and orientation errors by $46.78\%$ and $91.7\%$, respectively, compared with VDC, while the corresponding improvements over GIC are $12.21\%$ in position and $38.77\%$ in orientation. Likewise, the integral absolute position error $e_{\mathrm{IAE}}^p$ is reduced by $55.86\%$ relative to VDC and by $31.25\%$ relative to GIC. These improvements are obtained under comparable overall actuation levels, although VDC exhibits larger initial torque peaks, as also seen in Fig.~\ref{torques}.

Table~\ref{tab:controller_metrics} also reports the quantitative comparison between AMGC and MGC under parametric uncertainty. In both uncertainty cases, AMGC provides a clear improvement by effectively compensating for modeling errors. For instance, it achieves a steady-state position error $e_{\mathrm{ss}}^p$ on the order of $1~\mathrm{mm}$, indicating strong uncertainty-rejection capability and high final accuracy. However, MGC also preserves closed-loop stability and maintains satisfactory orientation regulation relative to AMGC, thereby highlighting the inherent robustness of the overall control architecture. Taken together, these results provide consistent numerical evidence for the stability and performance advantages of the proposed method.

\section{Conclusion}
\label{Conclusion}
This paper proposed an adaptive modular geometric control framework for robotic manipulators evolving on \(SE(3)\). The manipulator was decomposed into rigid-body and joint modules, with each rigid body described through an Euler--Poincaré-type spatial dynamics on \(\mathfrak{se}(3)\). Logarithmic configuration errors, adjoint-based body velocity errors, and screw-induced joint constraints were used to construct the geometric controller compatible with recursive twist, acceleration, and wrench propagation. The interconnection among modules was characterized through local power-pairing terms induced by the natural duality between \(\mathfrak{se}(3)\) and \(\mathfrak{se}(3)^*\). These virtual effort--flow terms cancel at the system level, enabling the module-level Lyapunov inequalities to establish exponential tracking stability of the complete nominal manipulator. The framework was further extended to uncertain dynamics through a geometric adaptation law on the manifold of symmetric positive-definite matrices \(\mathscr{P}(4)\), which preserves physical consistency of the estimates and yields semi-global uniform ultimate boundedness under parametric uncertainty. Simulations on a redundant high-inertia 4R manipulator demonstrated accurate tracking, smooth transients, and robustness to inertial uncertainty. In the non-adaptive case, the proposed controller reduced the RMS position error by \(46.78\%\) relative to the modular benchmark and by \(12.21\%\) relative to the geometric impedance benchmark, while maintaining comparable control effort. Future work will focus on experimental validation and extension to contact-rich and hybrid manipulation tasks.

\bibliographystyle{elsarticle-num}
\bibliography{bibl}

@techreport{bullo_murray_1995,
  author       = {Bullo, Francesco and Murray, Richard M.},
  title        = {Proportional Derivative (PD) Control on the Euclidean Group},
  institution  = {California Institute of Technology},
  year         = {1995},
  type         = {CDC Technical Report 95-010},
  note         = {Available electronically via \url{https://www.cds.caltech.edu/~murray/preprints/cds95-010.pdf}}
}

@book{murray2017mathematical,
  title={A mathematical introduction to robotic manipulation},
  author={Murray, Richard M and Li, Zexiang and Sastry, S Shankar},
  year={2017},
  publisher={CRC press}
}

@book{featherstone2008rigid,
  title={Rigid body dynamics algorithms},
  author={Featherstone, Roy},
  year={2008},
  publisher={Springer}
}

@article{park1995lie,
  title={A Lie group formulation of robot dynamics},
  author={Park, Frank C and Bobrow, James E and Ploen, Scott R},
  journal={The International journal of robotics research},
  volume={14},
  number={6},
  pages={609--618},
  year={1995},
  publisher={Sage Publications Sage CA: Thousand Oaks, CA}
}

@inproceedings{lee2018natural,
  title={A natural adaptive control law for robot manipulators},
  author={Lee, Taeyoon and Kwon, Jaewoon and Park, Frank C},
  booktitle={2018 IEEE/RSJ International Conference on Intelligent Robots and Systems (IROS)},
  pages={1--9},
  year={2018},
  organization={IEEE}
}

@book{zhu2010virtual,
  title={Virtual decomposition control: toward hyper degrees of freedom robots},
  author={Zhu, Wen-Hong},
  volume={60},
  year={2010},
  publisher={Springer Science \& Business Media}
}

@inproceedings{lee2010geometric,
  title={Geometric tracking control of a quadrotor UAV on SE (3)},
  author={Lee, Taeyoung and Leok, Melvin and McClamroch, N Harris},
  booktitle={49th IEEE conference on decision and control (CDC)},
  pages={5420--5425},
  year={2010},
  organization={IEEE}
}

@article{seo2023geometric,
  title={Geometric impedance control on SE (3) for robotic manipulators},
  author={Seo, Joohwan and Prakash, Nikhil Potu Surya and Rose, Alexander and Choi, Jongeun and Horowitz, Roberto},
  journal={IFAC-PapersOnLine},
  volume={56},
  number={2},
  pages={276--283},
  year={2023},
  publisher={Elsevier}
}

@article{seo2023contact,
  title={Contact-rich se (3)-equivariant robot manipulation task learning via geometric impedance control},
  author={Seo, Joohwan and Prakash, Nikhil PS and Zhang, Xiang and Wang, Changhao and Choi, Jongeun and Tomizuka, Masayoshi and Horowitz, Roberto},
  journal={IEEE Robotics and Automation Letters},
  volume={9},
  number={2},
  pages={1508--1515},
  year={2023},
  publisher={IEEE}
}

@article{lee2018geometric,
  title={A geometric algorithm for robust multibody inertial parameter identification},
  author={Lee, Taeyoon and Park, Frank C},
  journal={IEEE Robotics and Automation Letters},
  volume={3},
  number={3},
  pages={2455--2462},
  year={2018},
  publisher={IEEE}
}

@article{lee2019geometric,
  title={Geometric robot dynamic identification: A convex programming approach},
  author={Lee, Taeyoon and Wensing, Patrick M and Park, Frank C},
  journal={IEEE Transactions on Robotics},
  volume={36},
  number={2},
  pages={348--365},
  year={2019},
  publisher={IEEE}
}

@phdthesis{2019geometric,
  title={Geometric Methods for Dynamic Model-Based Identification and Control of Multibody Systems},
  author={TAEYOON LEE},
  year={2019},
  school={Department of mechanical and Aerospace Engineering, Seoul National University}
}

@article{humaloja2021decentralized,
  title={Decentralized observer design for virtual decomposition control},
  author={Humaloja, Jukka-Pekka and Koivum{\"a}ki, Janne and Paunonen, Lassi and Mattila, Jouni},
  journal={IEEE Transactions on Automatic Control},
  volume={67},
  number={5},
  pages={2529--2536},
  year={2021},
  publisher={IEEE}
}

@article{alcan2025constrained,
  title={Constrained trajectory optimization on matrix lie groups via lie-algebraic differential dynamic programming},
  author={Alcan, Gokhan and Abu-Dakka, Fares J and Kyrki, Ville},
  journal={Systems \& Control Letters},
  volume={204},
  pages={106220},
  year={2025},
  publisher={Elsevier}
}

@inproceedings{teng2022lie,
  title={Lie algebraic cost function design for control on Lie groups},
  author={Teng, Sangli and Clark, William and Bloch, Anthony and Vasudevan, Ram and Ghaffari, Maani},
  booktitle={2022 IEEE 61st Conference on Decision and Control (CDC)},
  pages={1867--1874},
  year={2022},
  organization={IEEE}
}

@article{zhu2011virtual,
  title={Virtual decomposition control for modular robot manipulators},
  author={Zhu, Wen-Hong and Vukovich, George},
  journal={IFAC Proceedings Volumes},
  volume={44},
  number={1},
  pages={13486--13491},
  year={2011},
  publisher={Elsevier}
}

@article{koivumaki2022subsystem,
  title={Subsystem-based control with modularity for strict-feedback form nonlinear systems},
  author={Koivum{\"a}ki, Janne and Humaloja, Jukka-Pekka and Paunonen, Lassi and Zhu, Wen-Hong and Mattila, Jouni},
  journal={IEEE Transactions on Automatic Control},
  volume={68},
  number={7},
  pages={4336--4343},
  year={2022},
  publisher={IEEE}
}

@article{hejrati2025impact,
  title={Impact-resilient orchestrated robust controller for heavy-duty hydraulic manipulators},
  author={Hejrati, Mahdi and Mattila, Jouni},
  journal={IEEE/ASME Transactions on Mechatronics},
  year={2025},
  publisher={IEEE}
}

@article{seo2025se,
  title={SE (3)-Equivariant robot learning and control: a tutorial survey},
  author={Seo, Joohwan and Yoo, Soochul and Chang, Junwoo and An, Hyunseok and Ryu, Hyunwoo and Lee, Soomi and Kruthiventy, Arvind and Choi, Jongeun and Horowitz, Roberto},
  journal={International Journal of Control, Automation and Systems},
  volume={23},
  number={5},
  pages={1271--1306},
  year={2025},
  publisher={Springer}
}

@inproceedings{abu2020geometry,
  title={Geometry-aware dynamic movement primitives},
  author={Abu-Dakka, Fares J and Kyrki, Ville},
  booktitle={2020 IEEE International Conference on Robotics and Automation (ICRA)},
  pages={4421--4426},
  year={2020},
  organization={IEEE}
}

@article{saveriano2023learning,
  title={Learning stable robotic skills on Riemannian manifolds},
  author={Saveriano, Matteo and Abu-Dakka, Fares J and Kyrki, Ville},
  journal={Robotics and Autonomous Systems},
  volume={169},
  pages={104510},
  year={2023},
  publisher={Elsevier}
}

@article{nah2025modular,
  title={Modular robot control with motor primitives},
  author={Nah, Moses C and Lachner, Johannes and Hogan, Neville},
  journal={arXiv preprint arXiv:2505.10694},
  year={2025}
}

@inproceedings{hejrati2025desired,
  title={Desired Impedance Allocation for Robotic Manipulators},
  author={Hejrati, Mahdi and Mattila, Jouni},
  booktitle={2025 IEEE 64th Conference on Decision and Control (CDC)},
  pages={2634--2641},
  year={2025},
  organization={IEEE}
}

@article{barjini2025surrogate,
  title={Surrogate-Enhanced Modeling and Adaptive Modular Control of All-Electric Heavy-Duty Robotic Manipulators},
  author={Barjini, Amir Hossein and Bahari, Mohammad and Hejrati, Mahdi and Mattila, Jouni},
  journal={arXiv preprint arXiv:2508.06313},
  year={2025}
}

@article{hejrati2025robust,
  title={Robust Immersive Bilateral Teleoperation of Beyond-Human-Scale Systems with Enhanced Transparency and Sense of Embodiment},
  author={Hejrati, Mahdi and Mustalahti, Pauli and Mattila, Jouni},
  journal={arXiv preprint arXiv:2505.14486},
  year={2025}
}

@article{ding2022vdc,
  title={VDC-based admittance control of multi-DOF manipulators considering joint flexibility via hierarchical control framework},
  author={Ding, Liang and Xing, Hongjun and Gao, Haibo and Torabi, Ali and Li, Weihua and Tavakoli, Mahdi},
  journal={Control Engineering Practice},
  volume={124},
  pages={105186},
  year={2022},
  publisher={Elsevier}
}

@article{rashad2022energy,
  title={Energy aware impedance control of a flying end-effector in the port-Hamiltonian framework},
  author={Rashad, Ramy and Bicego, Davide and Zult, Jelle and Sanchez-Escalonilla, Santiago and Jiao, Ran and Franchi, Antonio and Stramigioli, Stefano},
  journal={IEEE transactions on robotics},
  volume={38},
  number={6},
  pages={3936--3955},
  year={2022},
  publisher={IEEE}
}

@article{kumar2023tracking,
  title={Tracking control design for fractional order systems: A passivity-based port-Hamiltonian framework},
  author={Kumar, Lalitesh and Dhillon, Sukhwinder Singh},
  journal={ISA transactions},
  volume={138},
  pages={1--9},
  year={2023},
  publisher={Elsevier}
}

@article{dirksz2011port,
  title={Port-Hamiltonian and power-based integral type control of a manipulator system},
  author={Dirksz, Daniel A and Scherpen, Jacquelien MA},
  journal={IFAC Proceedings Volumes},
  volume={44},
  number={1},
  pages={13450--13455},
  year={2011},
  publisher={Elsevier}
}

@article{van2014port,
  title={Port-Hamiltonian systems theory: An introductory overview},
  author={Van Der Schaft, Arjan and Jeltsema, Dimitri},
  journal={Foundations and Trends{\textregistered} in Systems and Control},
  volume={1},
  number={2-3},
  pages={173--378},
  year={2014},
  publisher={Emerald Publishing Limited Boston—Delft}
}

@article{huang2026switching,
  title={Switching model predictive control considering nonlinear dynamics of long hydraulic pipelines and counterbalance valves for trajectory tracking of large-scale hydraulic manipulators},
  author={Huang, Weidi and Wang, Yaoxing and Cheng, Min and Qi, Hongli and Yang, Shuwei and Jia, Ruiheng and Xu, Bing},
  journal={ISA transactions},
  year={2026},
  publisher={Elsevier}
}

@article{guo2026adaptive,
  title={Adaptive RISE Control of Hydraulic Manipulators Using Actor-Critic Architecture},
  author={Guo, Yujie and Ai, Chao and Chen, Junxiang and Kong, Xiangdong},
  journal={IEEE Transactions on Automation Science and Engineering},
  year={2026},
  publisher={IEEE}
}

@book{krstic1995nonlinear,
  title={Nonlinear and adaptive control design},
  author={Krstic, Miroslav and Kokotovic, Petar V and Kanellakopoulos, Ioannis},
  year={1995},
  publisher={John Wiley \& Sons, Inc.}
}

@article{slotine1987adaptive,
  title={On the adaptive control of robot manipulators},
  author={Slotine, Jean-Jacques E and Li, Weiping},
  journal={The international journal of robotics research},
  volume={6},
  number={3},
  pages={49--59},
  year={1987},
  publisher={Sage Publications Sage CA: Thousand Oaks, CA}
}

@book{bullo2005geometric,
  title={Geometric control of mechanical systems: modeling, analysis, and design for simple mechanical control systems},
  author={Bullo, Francesco and Lewis, Andrew D},
  volume={49},
  year={2005},
  publisher={Springer}
}

@article{bullo1999tracking,
  title={Tracking for fully actuated mechanical systems: a geometric framework},
  author={Bullo, Francesco and Murray, Richard M},
  journal={Automatica},
  volume={35},
  number={1},
  pages={17--34},
  year={1999},
  publisher={Elsevier}
}

@article{murray1997nonlinear,
  title={Nonlinear control of mechanical systems: A Lagrangian perspective},
  author={Murray, Richard M},
  journal={Annual Reviews in Control},
  volume={21},
  pages={31--42},
  year={1997},
  publisher={Elsevier}
}

@book{lewis1995aspects,
  title={Aspects of geometric mechanics and control of mechanical systems},
  author={Lewis, Andrew David},
  year={1995},
  publisher={California Institute of Technology}
}

@article{zhou2026semi,
  title={Semi-global exponential convergence nonlinear control of the quadrotor unmanned aerial vehicle based on Lie algebra},
  author={Zhou, Zhihao and Xian, Bin and Cai, Jiaming and Yu, Jian},
  journal={ISA transactions},
  year={2026},
  publisher={Elsevier}
}

@book{marsden1999introduction,
  title={Introduction to mechanics and symmetry: a basic exposition of classical mechanical systems},
  author={Marsden, Jerrold E and Ratiu, Tudor S and Golubitsky, M},
  volume={17},
  year={1999},
  publisher={Springer}
}

\end{document}